\documentclass[11pt]{article}
\usepackage{amssymb,amsmath,amsthm,amsfonts}
\usepackage{graphicx}
\usepackage{subfigure}
\usepackage[usenames, dvipsnames]{color}
\usepackage{mathrsfs}
\usepackage[numbers]{natbib}
\usepackage{epstopdf}
\usepackage[CJKbookmarks, colorlinks, bookmarksnumbered=true,pdfstartview=FitH,linkcolor=black]{hyperref}
\def\disp{\displaystyle}

\def\crr{\cr\noalign{\vskip2mm}}

\def\dref#1{(\ref{#1})}

\theoremstyle{plain}
\newtheorem{theorem}{Theorem}[section]
\newtheorem{lemma}{Lemma}[section]

\newtheorem{corollary}{Corollary}[section]

\numberwithin{equation}{section}

\theoremstyle{definition}
\newtheorem{definition}{Definition}
\newtheorem{assumption}{Assumption}[section]

\newtheorem{remark}{Remark}[section]

\newcommand{\R}{{\mathbb R}}

\def\A{\mathcal{A}}
\def\B{\mathcal{B}}

\begin{document}

\title{{\bf
 Actuator Dynamics Compensation in Stabilization of
Abstract Linear Systems\footnote{\small
This work is  supported by  the National Natural Science
Foundation of China (Nos. 61873153,  61873260). }} }

\author{ Hongyinping Feng$^{a}$\footnote{\small  Corresponding author.
Email: fhyp@sxu.edu.cn.}, \  Xiao-Hui Wu$^{a}$\ and \ Bao-Zhu Guo$^{b,c}$\\
$^a${\it \small School of Mathematical Sciences,}
{\it \small Shanxi University,  Taiyuan, Shanxi, 030006,
China}\\
$^b${\it \small Department of Mathematics and Physics,  }
{\it \small  North
China Electric Power University, Beijing 102206, China}\\
$^c${\it \small Key Laboratory of System and Control, Academy of Mathematics and Systems Science,}
\\{\it \small  Academia Sinica, Beijing, China}
 }

 \maketitle

\begin{abstract}
This is the first part of four series papers, aiming at the problem of  actuator dynamics compensation for  linear systems.
We consider the stabilization of  a type of cascade abstract linear systems
 which model the actuator dynamics compensation
 for linear systems where both the control plant and  its actuator dynamics can be infinite-dimensional.
We develop a systematic way to stabilize the cascade systems by a full state feedback.
Both the well-posedness and the exponential stability of the resulting closed-loop system are established in the abstract framework.
 A sufficient condition of the existence of   compensator for ordinary differential equation  (ODE) with partial differential equation (PDE) actuator dynamics is obtained.
 The feedback design is based on a novelly constructed upper-block-triangle transform and the Lyapunov function design   is not  needed in the stability analysis.
As applications,  an  ODE  with  input  delay and an unstable heat equation with  ODE actuator dynamics are investigated to validate the  theoretical results.
 The numerical simulations for the unstable heat system are carried out to validate the proposed
approach visually.
\vspace{0.3cm}


\vspace{0.3cm}

\end{abstract}

\section{Introduction}
System control through actuator dynamics can usually be modeled  as a cascade control  system which  has been intensively investigated  in the last  decade.
An early infinite-dimensional actuator dynamic compensation is the  input  time-delay compensation for finite-dimensional  systems in the name of the Smith predictor (\cite{Smith1959ISA})  and its modifications (\cite{Artstein1982TAC,ManitiusOlbrot1979TAC}).
In  \cite{SmyshlyaevKrstic2008SCL},   the partial differential equation (PDE)  backstepping method was  developed to cope with the time-delay problem.
Regarding the time-delay as the dynamics dominated by a
  transport  equation,
 the   input delays  compensation problem comes down to    the  boundary  control  of an ODE-PDE cascade.
Actually, the PDE backstepping method   can  compensate for  various actuator dynamics
which include  but not limited to  the general first order
hyperbolic equation  dynamics \cite{SmyshlyaevKrstic2008SCL},    the heat equation dynamics  \cite{Krstic2009SCL,Sustoa2010JFI,TangSCL2011,Wang2015AUT},     the wave equation dynamics
\cite{Krstic2009TAC,Sustoa2010JFI} and the Schr\"{o}dinger equation dynamics
\cite{WangJMSCL2013}.
 However,  the PDE backstepping transformation relies strongly  on the choice of
the  target systems which are built on the basis of intuition not theory.
This implies that an  inappropriate target system may make the PDE backstepping  method not be always
working.
What is more, since the kernel function of the backstepping transformation  is usually governed by a PDE,
there are some formidable  difficulties for PDE backstepping method in dealing with   some infinite-dimensional dynamics like
those dominated by the Euler-Bernoulli beam equations, multi-dimensional PDEs, and even  the one dimensional  PDEs with variable coefficients.

In this paper, we propose a systematic and generic  way to deal with   the
actuator dynamics compensation by stabilizing an abstract cascade linear system.
The central effort focuses on the unification of various  actuator dynamics compensations
from a general abstract framework point of view.
The problem is  described by the following system:
\begin{equation} \label{201911101138}
 \left\{\begin{array}{l}
\disp \dot{x}_1(t) = A_1 x_1(t)+B_1C_2x_2(t) ,\ \ t>0,\crr
\disp  \dot{x}_2(t) = A_2 x_2(t) +B_2 u(t),\ \ t>0,
\end{array}\right.
\end{equation}
 where $A_1:X_1\to X_1$ is the operator of   control plant, $A_2:X_2\to X_2$ is the operator of    actuator dynamics, $B_1C_2 :X_2\to X_1$  is the interconnection between the control plant and its control dynamics, $B_2:U_2\to X_2  $ is the  control operator, and   $u(t)$ is the  control.
 The state space and the control space are $ X_1\times X_2$ and $U_2$, respectively. All the operators appeared in \dref{201911101138} can be unbounded.
The main objective  of this paper is to  seek a state   feedback   to stabilize  the  abstract system \dref{201911101138} exponentially.  We limit ourselves  to the full state feedback because, thanks to the separation principle of   linear systems, the output feedback law is straightforward  once the state  observer  of system \dref{201911101138} is available.
The observer design with  sensor dynamics  would be the  next paper \cite{FPart2} of this series of studies  before the last part on the control of uncertain systems \cite{FPart4}.

It is well known that the   cascade system can be  decoupled by  a block-upper-triangular transformation which is related to a Sylvester operator equation \cite{Sylvester1991}.
This inspires us to stabilize the cascade
  system
 by decoupling
   the cascade system first and then stabilizing  the  decoupled system.
The system decoupling   needs to solve  the Sylvester operator equations which may be a difficult task   particularly when the corresponding operators are unbounded.
Fortunately, the problem becomes relatively easy  provided at least one of $A_1$ and $A_2$ is bounded.  In this way, numerous  actuator dynamics dominated by the transport  equation \cite{SmyshlyaevKrstic2008SCL}, wave equation \cite{Krstic2009TAC},    heat equation  \cite{Krstic2009SCL}
 as well as  the  Euler-Bernoulli beam equation \cite{FengWubeam}
 can be treated in a unified way.
In this paper, this fact will be demonstrated through  two different  systems:
an  ODE  system with  input   delay  and an unstable heat system with ODE actuator dynamics.
We point out that the considered  stabilization of  heat-ODE cascade system is an  interesting  and challenging problem  because the actuator dynamics is finite-dimensional yet  the  control plant is of the infinite-dimension. In other words, what we need to do is to  control  an ``infinite-dimensional'' system via  a ``finite-dimensional'' compensator.
Compared  with the ODE system with PDE  actuator dynamics, the results about  PDE system with ODE actuator  dynamics are  still   scarce.

The rest of the paper is organized as follows:
In Section \ref{Se.2}, we demonstrate  the  main idea through   an ODE cascade system.
Sections \ref{Se.3} and \ref{Se.4}  give  some  preliminary results  about the similarity of  operators  and the Sylvester equation.
Section \ref{Se.5} is devoted to the dynamics compensator design.
The well-posedness and the exponential stability are also established. In Section  \ref{Se.6}, we apply  the proposed method to the input  delay compensation for an ODE system.
  Stabilization of  an unstable heat equation by  finite-dimensional
  actuator dynamics is considered in  Section \ref{Se.7}.
 Section \ref{Se.8} presents some numerical simulations,
 followed up conclusions in Section \ref{Se.9}. For the sake of readability,
  some results that are less relevant to the dynamics compensator
design are arranged  in the Appendix.

Throughout the paper,
the identity operator  on the Hilbert space $X_i$  will be denoted by $I_i,\  i=1,2$, respectively.
The space of bounded linear operators from $X_1$ to $X_2$ is denoted by $\mathcal{L}(X_1, X_2 )$.
The spectrum, resolvent set, the range, the kernel and the domain  of
the  operator $A$ are  denoted by $\sigma(A)$, $\rho(A)$,
${\rm Ran}(A)$, ${\rm Ker} (A)$ and $D(A)$, respectively.
The transpose of matrix $A$  is denoted   by $A^\top$.

\section{Finite-dimensional dynamics }\label{Se.2}
In order to introduce  our main idea  clearly,  we first consider system \dref{201911101138} in the finite-dimensional case.
Suppose that   $X_1, X_2$, $U_1$   and $U_2$  are the Euclidean spaces,  $C_2:X_2\to U_1$ and  $B_1:U_1\to X_1$.
We shall  design a full state feedback to stabilize the cascade system \dref{201911101138}.
Although this problem can be  achieved completely by
the pole assignment theorem   provided  system \dref{201911101138}  is controllable,
when we come across that at least one of
  $A_1$ and  $A_2$ is an operator in   an infinite-dimensional Hilbert space,
the problem would  become very   complicated.
We thus need an alternative  treatment    that can be   extended to the setting of  infinite-dimensional framework.

We first divide the controller into two parts:
\begin{equation} \label{20191017804}
\left.\begin{array}{l}
\disp u(t)=K_2x_2(t)+u_c(t),
\end{array}\right.
\end{equation}
where $K_2\in \mathcal{L}(X_2, U_2)$ is chosen to make  $A_2+B_2K_2$ Hurwitz  and $u_c(t)$ is a new control to be designed.
The main role played by the first part $K_2x_2(t)$ of the controller is  to stabilize the $x_2$-subsystem.
If $(A_2,B_2)$ is controllable, such a $K_2$ always exists.
Under \dref{20191017804}, the control plant \dref{201911101138} becomes
\begin{equation} \label{20191024826}
\left\{\begin{array}{l}
\disp \dot{x}_1(t) = A_1 x_1(t)+B_1C_2x_2(t),\crr
\disp  \dot{x}_2(t) = (A_2+B_2 K_2) x_2(t) +B_2 u_c(t).
\end{array}\right.
\end{equation}
Now, we   decouple system \dref{20191024826} by  the
block-upper-triangular transformation:
\begin{equation} \label{20191024837}
\begin{pmatrix}
I_1&S\\
0&I_2
\end{pmatrix}
\begin{pmatrix}
A_1&B_1C_2 \\
0&A_2+B_2K_2
\end{pmatrix}\begin{pmatrix}
I_1& S\\
0&I_2
\end{pmatrix}^{-1}
=
\begin{pmatrix}
A_1&S(A_2+B_2K_2)-A_1S+B_1C_2 \\
0&A_2+B_2K_2
\end{pmatrix},
\end{equation}
where $S\in\mathcal{ L}(X_2,X_1)$ is to be determined.
Evidently, the  system matrix  of \dref{20191024826} is   block-diagonalized  if  the matrix $S$ solves the matrix equation
  \begin{equation} \label{20191024842}
A_1S-S(A_2+B_2K_2)=B_1C_2 ,
\end{equation}
which is a well known Sylvester equation. An immediate consequence of this fact is that
  the controllability of the
following pairs is  equivalent:
\begin{equation} \label{201910171043}
\left(\begin{pmatrix}
A_1& B_1C_2  \\
0&A_2+B_2K_2
\end{pmatrix},\
 \begin{pmatrix}
0\\
B_2
\end{pmatrix}\right)\ {\rm  and }\
\left(\begin{pmatrix}
A_1&0\\
0&A_2+B_2K_2
\end{pmatrix},\
\begin{pmatrix}
SB_2 \\
B_2
\end{pmatrix}
\right).
\end{equation}
Owing to the block-diagonal structure, the stabilization of the second  system of \dref{201910171043} is
much   easier than the first one. Indeed, since $A_2+B_2K_2$ is stable already,
  the controller  $u_c(t)$ in \dref{20191024826} can be designed  by   stabilizing  system $(A_1,SB_2)$ only:
  \begin{equation} \label{201910241002}
u_c(t)=(K_1,0)\begin{pmatrix}
I_1&S\\
0&I_2
\end{pmatrix}
\begin{pmatrix}
x_1(t)\\
x_2(t)
\end{pmatrix}
=K_1Sx_2(t)+K_1x_1(t),
\end{equation}
  where
$K_1\in \mathcal{L}(X_1,U_2 )$   is chosen to make  $A_1+SB_2K_1$ Hurwitz.
 In view of  \dref{20191017804},  the   controller of the original system \dref{201911101138} is therefore designed as
 \begin{equation} \label{201911281134}
u(t)
=K_2x_2(t)+K_1x_1(t)+K_1Sx_2(t),
\end{equation}
which leads to the closed-loop
 of system \dref{201911101138}:
 \begin{equation} \label{201912061014}
 \left\{\begin{array}{l}
\disp \dot{x}_1(t) = A_1 x_1(t)+B_1C_2 x_2(t),\crr
\disp  \dot{x}_2(t) =   (A_2 +B_2K_2+B_2 K_1S)x_2(t) +B_2 K_1{x}_1(t).
\end{array}\right.
\end{equation}

   \begin{lemma}\label{Lm20191017836}
   Let    $X_1, X_2$, $U_1$ and $ U_2$ be Euclidean spaces, and let
   $A_j\in \mathcal{L}(X_j)$,   $B_j \in  \mathcal{L}(U_j ,X_j)$, $C_2 \in \mathcal{L}(X_2,U_1)$  and
 $K_j\in \mathcal{L}(X_j,U_2)$, $j=1,2$. Suppose that $A_2 +B_2K_2$ and $A_1+SB_2K_1$ are Hurwitz,
 and $S\in \mathcal{L}(X_2,X_1)$
    is the solution  of the
Sylvester equation \dref{20191024842}.
 Then,
 the closed-loop system \dref{201912061014}  is  stable in $X_1\times X_2$.
\end{lemma}
\begin{proof}
Let
\begin{equation} \label{201910201037A1}
 \A_1= \begin{pmatrix}
A_1& B_1C_2  \\
 B_2K_1&A_2 +B_2K_2+B_2 K_1S
\end{pmatrix}
\end{equation}
and
\begin{equation} \label{201910201037A2}
\A_2=
 \begin{pmatrix}
A_1+SB _2K_1& 0 \\
 B_2 K_1&A_2+B_2K_2
\end{pmatrix}.
\end{equation}
Then, the matrices $\A_1$ and $\A_2$ are similar each other.
Since both $A_2 +B_2K_2$ and $A_1+SB_2K_1$ are Hurwitz,   $\A_2$ is  Hurwitz and hence
$\A_1$ is  Hurwitz as well. This completes the  proof.
\end{proof}

To sum up,  the scheme  of the actuator dynamics compensator  design for  system \dref{201911101138} consists of three steps: (a)~ Find  $K_2\in \mathcal{L}(X_2,U_2 )$ to stabilize  system $(A_2,B_2)$; (b)~ Solve the  Sylvester equation \dref{20191024842}; (c)~ Find  $K_1\in \mathcal{L}(X_1,U_2 )$ to stabilize system $(A_1,SB_2)$.  By the  pole assignment,  the   step (a) is almost straightforward. When the  solution  of  Sylvester  equation \dref{20191024842} is available,
the  step (c) is  straightforward too.  As for the   step (b), we have the following Lemma \ref{Lm201910142119}.
\begin{lemma}\label{Lm201910142119}(\cite{Rosenblum1956DMJ})
 Let    $X_1, X_2$,   $ U_1$ and $U_2$ be  Euclidean spaces, and let $A_j\in \mathcal{L}(X_j)$,
  $B_j \in  \mathcal{L}(U_j ,X_j)$,  $C_2 \in \mathcal{L}(X_2,U_1)$
   and $K_2\in \mathcal{L}(X_2,U_2 )$, $j=1,2$. Suppose that
\begin{equation} \label{20191017834}
\sigma(A_1)\cap\sigma(A_2+B_2K_2 )=\emptyset.
\end{equation}
Then, the Sylvester equation \dref{20191024842}
 admits a unique  solution
$S\in\mathcal{ L}(X_2,X_1)$ which  can be represented as
    \begin{equation} \label{201910171852}
S=\frac{1}{2\pi i}\int_{\Gamma}(A_1-\lambda   )^{-1}B_1C_2 (A_2+B_2K_2-\lambda  )^{-1}d\lambda,
\end{equation}
where
$\Gamma$ is a smooth contour around $\sigma(A_1)$ and separated from $\sigma(A_2+B_2K_2)$,
with positive orientation.
\end{lemma}

\section{Preliminaries on abstract systems }\label{Se.3}
In order  to extend the   results in Section \ref{Se.2} to  infinite-dimensional systems,
 we  present  some preliminary background  on abstract infinite-dimensional  systems.
We  first  introduce the  dual   space    with
respect to a pivot space,  which has been discussed extensively in \cite{TucsnakWeiss2009book} particularly
for those systems with unbounded control and observation operators.

Suppose that $X$ is a Hilbert space and  $A : D(A )\subset X \to X $ is a densely defined operator  with $\rho(A) \neq \emptyset$. Then, $A$ can determine two Hilbert spaces:
$(D(A), \|\cdot\|_1)$ and $([D(A^*)]', \|\cdot\|_{-1})$, where
$ [D(A^*)]' $ is the dual space of $ D(A^*) $ with respect to  the  pivot space $X$, and the norms $\|\cdot\|_1$ and $\|\cdot\|_{-1}$ are defined   by
\begin{equation} \label{20191141722}
 \left\{\begin{array}{l}
 \disp \|x\|_1=\|(\beta -A)x\|_X,\ \ \forall\ x\in D(A),\crr
 \disp \|x\|_{-1}=\disp \|(\beta -A)^{-1}x\|_X,   \ \ \forall\ x\in X,
 \end{array}\right. \ \ \beta\in\rho(A).
\end{equation}
These two spaces are independent of the choice of $\beta\in\rho(A)$ since  different  choices
of $\beta$ lead  to  equivalent norms.
 For brevity, we denote
the two spaces   as
$ D(A) $ and $ [D(A^*)]' $ in the sequel.
 The adjoint of $A^*\in \mathcal{L}(D(A^*),X)$, denoted by  $\tilde{A}$, is defined   as
 \begin{equation} \label{20191121602}
 \disp     \langle \tilde{A} x, y\rangle_{ [D(A^*)]',  D(A^*)}=
 \langle x,  A ^*y\rangle_{X },\ \ \forall\ x\in X,  y\in D(A ^*).
\end{equation}
It is evident that $\tilde{A} x=Ax$ for any $x\in D(A)$. Hence,
  $\tilde{A}\in \mathcal{L}(X,  [D(A^*)]')$ is an extension  of $A $. Since $A$ is
densely defined,  the extension is unique. By \cite[Proposition 2.10.3]{TucsnakWeiss2009book}, for any $\beta\in\rho(A)$,  we have
  $(\beta  -\tilde{A})\in \mathcal{L}(X,[D(A^*)]')$  and
$(\beta -\tilde{A})^{-1}\in \mathcal{L}( [D(A^*)]',X)$ which  imply  that
$\beta -\tilde{A}$ is an isomorphism from $X$ to $[D(A^*)]'$.
If $A$ generates a $C_0$-semigroup $e^{At}$ on $X$, then,  so is for its extension $\tilde{A}$ and
$e^{\tilde{A}t}=  (\beta -\tilde{A})  e^{At}(\beta -\tilde{A})^{-1}$.

Suppose that $Y$ is an output Hilbert space  and  $C\in \mathcal{L}(D(A),Y)$. The
$\Lambda$-extension of $C $  with respect to $A $  is defined by
   \begin{equation} \label{20205281601}
 \left\{ \begin{array}{l}
 \disp C_{ \Lambda}x=\lim\limits_{\lambda\rightarrow +\infty}
C \lambda(\lambda    - {A} )^{-1}x,\ \ x\in  D (C_{ \Lambda}),\crr
\disp D (C_{ \Lambda})=\{x \in X \ |\  \mbox{the above limit exists}\}.
\end{array}\right.
\end{equation}
 Define the norm
\begin{equation} \label{20205281609}
  \|x\|_{D(C_{\Lambda})}= \|x\|_X+\sup_{\lambda\geq\lambda_0}
  \|C \lambda(\lambda  - {A} )^{-1}x\|_Y,\ \ \forall\ x\in D(C_{\Lambda}),
\end{equation}
where $\lambda_0\in \R$  satisfies   $[\lambda_0, \infty ) \subset \rho(A) $.
By  \cite[Proposition 5.3]{Weiss1994MCSS},   $D(C_{ \Lambda})$ with the norm   $\|\cdot\|_{D(C_{\Lambda})}$ is
   a Banach space and
$C_{ \Lambda}\in \mathcal{L}(  D(C_{ \Lambda}),Y)$.
Moreover,
  there exist continuous embeddings:
 \begin{equation} \label{20205281604}
 D(A) \hookrightarrow  D(C_{ \Lambda})\hookrightarrow  X   \hookrightarrow [ D (A ^*)]'.
\end{equation}

 For other concepts of the
admissibility for both control and observation operators, and the regular linear systems,
we refer to \cite{Weiss1989SICON,Weiss1994MCSS,Weiss1989}.

\begin{definition}\label{De20191122107}
Suppose that  $X$  is  a Hilbert space and
  $A_j: D(A_j)\subset X\to X$ is a densely defined operator with $\rho(A_j)\neq\emptyset$,
  $j=1,2$.
  We say that the operators   $ A_1 $  and  $ A_2 $ are  similar     with the
 transformation $P$, denoted by $ A_1 \sim_{P}  A_2 $,   if
  the  operator  $P \in \mathcal{L}(X)$  is   invertible   and satisfies
   \begin{equation} \label{201911281703}
 \left.\begin{array}{l}
\disp PA_1P^{-1}=A_2 \ \ \mbox{and}\ \  D(A_2) = P D(A_1).
\end{array}\right.
\end{equation}
   \end{definition}
  Suppose that   $ A_1 \sim_{P}  A_2 $.
 Then,  $ A_1^*\sim_{P^{-1*}} A_2^*$ and
 in  particular, $P^*D(A_2^*)=D(A_1^*)$.
Obviously,
  $ A_1 \sim_{P}  A_2 $ implies that
 $A_1   $  generates a $C_0$-semigroup
 $e^{A_1t}$ on  $X$ if and only if $A_2  $  generates a $C_0$-semigroup
 $e^{A_2t}$ on  $X$. More specifically,  $ P e^{A_1t}  P^{-1}=e^{A_2t}$.

%
%
%
%

\begin{lemma}\label{th201911262002}
 Let $X$ and  $U$ be   Hilbert spaces.
Suppose that   the operator $A_j : D(A_j )\subset X \to X $  generates a $C_0$-semigroup
 $e^{A_jt}$ on  $X$ and  $B_j\in \mathcal{L}(U,[D(A_j^*)]')$, $j=1,2$.
 If   $ A_1 \sim_{P}  A_2 $ and
 \begin{equation} \label{2018820915}
 \left.\begin{array}{l}
\disp \langle B_2u,x\rangle_{ [D(A_2^*)]', D(A_2^*)}=
\langle B_1u,P^*x\rangle_{[D(A_1^*)]',D(A_1^*)},
 \ \ \forall\  \ u\in U, x\in D(A_2^*),
\end{array}\right.
\end{equation}
   then,   the following assertions hold true:

(i).  $B_1$  is admissible for $e^{A_1 t}$ if and only if $ B_2$
 is admissible for  $e^{A_2t}$;

 (ii).   $(A_1,B_1)$ is exactly (or approximately) controllable if and only if   $(A_2,  B_2)$ is exactly (or approximately)  controllable.

\end{lemma}

\begin{proof}
  We first prove (i).
  For any $f\in L_{\rm  loc}^2([0,\infty);U)$ and $\phi\in D(A_2^*)$, it follows from \dref{2018820915} that
  \begin{equation} \label{20191128803}
 \begin{array}{l}
 \left \langle e^{\tilde{A}_2(t-s)}B_2f(s), \phi\right\rangle_{[D(A_2^*)]',D(A_2^*)}
\disp =  \left \langle B_2f(s), e^{ {A}_2^*(t-s)}\phi\right\rangle_{[D(A_2^*)]',D(A_2^*)}\crr
\disp =  \left \langle B_1f(s), P^*e^{ {A}_2^*(t-s)}\phi\right\rangle_{[D(A_1^*)]',D(A_1^*)}
\disp =  \left \langle B_1f(s), P^*e^{ {A}_2^*(t-s)}P^{*-1} P^*\phi\right\rangle_{[D(A_1^*)]',D(A_1^*)}\crr\disp
=  \left \langle B_1f(s),  e^{ {A}_1^*(t-s)}  P^*\phi\right\rangle_{[D(A_1^*)]',D(A_1^*)}
= \disp  \left \langle e^{ \tilde{A}_1 (t-s)}B_1f(s),  P^*\phi\right\rangle_{[D(A_1^*)]',D(A_1^*)}
 \end{array}
\end{equation}
for any $t>0$ and $0\leq s\leq t$.
Define the operator $\Phi_{j  } (t) \in \mathcal{L}( L_{\rm loc}^2([0,\infty);U),[D(A_j^*)]') $ by
\begin{equation} \label{20191128751}
 \begin{array}{l}
\disp \Phi_{j  } (t) f= \int_0^{t}e^{\tilde{A}_j(t-s)}B_jf(s)ds,\ \ \forall\ \ f\in L_{\rm loc}^2([0,\infty);U),\    \forall\  t>0, \ j=1,2.
\end{array}\end{equation}
Then,  it follows from \dref{20191128803} that
  \begin{equation} \label{20191122001}
 \begin{array}{l}
\disp \left \langle \Phi_{2} (t)f, \phi\right\rangle_{[D(A_2^*)]',D(A_2^*)}
= \int_0^t\left \langle e^{\tilde{A}_2(t-s)}B_2f(s), \phi\right\rangle_{[D(A_2^*)]',D(A_2^*)}ds\crr
 = \disp \int_0^t\left \langle e^{ \tilde{A}_1 (t-s)}B_1f(s),  P^*\phi\right\rangle_{[D(A_1^*)]',D(A_1^*)}ds
= \disp \left \langle  \int_0^t e^{ \tilde{A}_1 (t-s)}B_1f(s)ds,  P^*\phi\right\rangle_{[D(A_1^*)]',D(A_1^*)}\crr
\disp =   \left \langle  \Phi_1(t) f,  P^*\phi\right\rangle_{[D(A_1^*)]',D(A_1^*)}.
\end{array}
\end{equation}
When   $B_1$ is admissible for $e^{A_1 t}$, we have $\Phi_1(t) f\in X $ and thus
   \begin{equation} \label{20191128735}
P \Phi_1(t) f \in X ,\ \ \forall\ t>0.
\end{equation}
  Combining \dref{20191122001} and \dref{20191128735}, we arrive at
 \begin{equation} \label{20191128736}
 \begin{array}{l}
\disp \left \langle  \Phi_2(t) f , \phi\right\rangle_{[D(A_2^*)]',D(A_2^*)}=
 \left \langle  \Phi_1(t) f,  P^*\phi\right\rangle_{X}
=\disp \left \langle P \Phi_1(t) f,  \phi\right\rangle_{[D(A_2^*)]',D(A_2^*)},
\end{array}
\end{equation}
and hence,
\begin{equation} \label{20191128742}
 \Phi_2(t)f = P \Phi_1(t)f \ \ \mbox{in}\ \ [D(A_2^*)]', \ \forall\ t>0
\end{equation}
due  to   the arbitrariness  of $\phi$. \dref{20191128735} and \dref{20191128742} imply that
$  \Phi_2(t) f \in X$ for any $t>0$ which means that $B_2$ is admissible for $e^{A_2 t}$.

 When   $B_2$ is admissible for $e^{A_2t}$,  $\Phi_2(t) f\in X $ and thus
   \begin{equation} \label{20191128735B2}
P^{ -1} \Phi_2(t) f \in X ,\ \ \forall\ t>0.
\end{equation}
   Combining \dref{20191122001} and \dref{20191128735B2}, we have
 \begin{equation} \label{20191128736B2}
 \begin{array}{l}
\disp \left \langle  P^{ -1}\Phi_2(t) f , P^*\phi\right\rangle_{[D(A_1^*)]',D(A_1^*)}=
\left \langle  \Phi_2(t) f ,  \phi\right\rangle_{[D(A_2^*)]',D(A_2^*)}
= \disp \left \langle  \Phi_1(t) f,  P^*\phi\right\rangle_{[D(A_1^*)]',D(A_1^*)},
\end{array}
\end{equation}
and hence,
\begin{equation} \label{20191128742B2}
 P^{ -1} \Phi_2(t)f = \Phi_1(t)f\  \mbox{ in }\  [D(A_1^*)]', \ \  \forall\ t>0
\end{equation}
due  to   the arbitrariness  of $\phi$. \dref{20191128735B2} and \dref{20191128742B2} imply that
$  \Phi_1(t) f \in X$ for any $t>0$ and hence $B_1$ is admissible for $e^{A_1 t}$.

We next    prove (ii). By the proof of (i), the equality \dref{20191128742}
  always holds
provided  $B_1$ is admissible for $e^{A_1 t}$ or
$B_2$ is admissible for $e^{A_2t}$.
Since $P \in \mathcal{L}(X)$ is invertible,  we conclude that
 ${\rm Ran} ({\Phi_1(t)})=X$ if and only if ${\rm Ran}(\Phi_2(t))=X$
  ( or $\overline{{\rm Ran}( \Phi_1(t))}=X$ if and only if $\overline{{\rm Ran}(\Phi_2(t))}=X$ ).
 This completes the proof of the lemma.
\end{proof}

\begin{remark}\label{Re20191128742}
When $B_1 $ and $B_2$ are  bounded, \dref{2018820915}
implies that $B_2=PB_1$. In this case, systems $(A_1,B_1)$ and
 $(PA_1P^{-1},PB_1)$  have the same
  controllability, which is the same as the finite-dimensional counterpart.
\end{remark}

By  the separation principle  of the   linear systems,
 a fair amount of      closed-loop systems  resulting  from   the  observer based output feedback
can be converted into a cascade system.  The same thing also  takes place in the actuator dynamics compensation.
 At the end of this  section, we consider the well-posedness and stability of   general  cascade systems,
 which is  useful for the    well-posedness and stability
 analysis of
the closed-loop system.
 Moreover,  it   can   simplify and unify the  proofs  in \cite{FengGuoWu2019TAC,FengGuo2017TAC,Zhou2017JDE}.

 \begin{lemma}\label{Lm201912031723}
 Let $X_1,X_2$ and $U_1$ be Hilbert spaces.
 Suppose that  $A_j $   generates  a
 $C_0$-semigroup  $e^{A_jt}$ on $X_j$,
 $B_1\in \mathcal{L}(U_1,[D(A_1^*)]')$
 is admissible for $e^{A_1t}$ and
$C_2\in \mathcal{L}(  D(A_2) , U_1)$ is admissible for $e^{A_2t}$, $j=1,2$.
   Let
  \begin{equation} \label{2020428921}
 \A =\begin{pmatrix}
\tilde{A}_1&B_1 C_{2\Lambda} \\0&\tilde{A}_2
\end{pmatrix},
\left.\begin{array}{l}
 \disp D(\A )=\left\{(x_1,x_2) ^{\top}\in X_{ 1}\times D(A_2)  \ | \  \tilde{A}_1x_1 +B_1C_{2\Lambda}x_2\in X_1 \right\}.
 \end{array}\right.
 \end{equation}
 Then, the operator $\A $   generates a
 $C_0$-semigroup $e^{\A  t}$ on $X_1\times X_2$.
 In addition, if we suppose further that  $e^{A_jt}$ is
   exponentially stable in $X_j$, $j=1,2$, then, $e^{\A  t}$ is  exponentially stable in $X_1\times X_2$.
\end{lemma}
\begin{proof}
The operator $\A $ is associated with the following system:
\begin{equation}\label{20191311712FHAD}
\left\{\begin{array}{ll}
  \disp  \dot{x}_1(t)=A_1x_1(t)+B_1C_2x_2(t),\crr
   \disp  \dot{x}_2(t)=A_2x_2(t).
\end{array}\right.
\end{equation}
Since $x_2$-subsystem is independent of $x_1$-subsystem, for any  $(x_{1}(0),x_{2}(0)) ^{\top}\in D(\A)$,
we solve \dref{20191311712FHAD} to obtain  $x_2(t)=e^{A_2t}x_{2}(0)$.
Moreover, it follows from
  \cite[Proposition 2.3.5, p.30]{TucsnakWeiss2009book},
  \cite[Proposition 4.3.4, p.124]{TucsnakWeiss2009book}  and the admissibility of $C_2$ for $e^{A_2t}$ that
    $x_2 \in C^1([0,\infty); X_2)$
and
\begin{equation}\label{wxh20208121700}
C_{2}x_2(\cdot)=C_2e^{A_2\cdot}x_{2}(0)\in  {H}^1_{\rm loc} ([0,\infty); U_1).
\end{equation}
  Since $\tilde{A}_1x_{1}(0) +B_1C_2x_{2}(0)\in X_1$, by  the admissibility of $B_1$ for $e^{A_1t}$,   \cite[Proposition 4.2.10, p.120]{TucsnakWeiss2009book} and \dref{wxh20208121700}, it follows that
 the solution of  the $x_1$-subsystem   satisfies $x_1\in C^1([0,\infty); X_1)$.
Therefore, system \dref{20191311712FHAD} admits  a unique continuously differentiable  solution  $(x_1,x_2)^\top\in C^1([0,\infty);X_1\times X_2)$ for    any   $(x_{1}(0),x_{2}(0)) ^{\top}\in D(\A)$. By \cite[Theorem 1.3, p.102]{Pazy1983Book}, the operator $\A $   generates a  $C_0$-semigroup $e^{\A  t}$ on $X_1\times X_2$.

We next   show the exponential stability.  Suppose that $(x_1,x_2)^{\top}\in C^1 ([0,\infty);X_1\times X_2)$ is a  classical    solution  of
 system \dref{20191311712FHAD}.
  Since
 $e^{A_jt}$  is  exponentially stable in $X_j$, there exist  two positive constants
$\omega_j$ and $L_j$  such that
\begin{equation}\label{201994746FHAD}
\left\| e^{ A_j t} \right\|  \leq L_j e^{-\omega_j t},\ \ \forall\  t\geq0, \ \ j=1,2.
\end{equation}
Hence,
  \begin{equation}\label{201994835FHAD}
\|x_2(t)\|_{X_2}\leq L_2e^{-\omega_2t}\|x_2(0)\|_{X_2},\ \ \forall\ t\geq0.
\end{equation}
By  the admissibility of $C_2$ for $e^{A_2 t}$   and  \cite[Proposition 4.3.6, p.124]{TucsnakWeiss2009book}, it follows that
     \begin{equation}\label{201994731FHAD}
 v_{\omega}\in L^2([0, \infty );U_1),\ \ v_{\omega}(t)= e^{\omega t}C_2x_2(t),\ \ 0<\omega<\omega_2  .
\end{equation}
We combine   \cite[Remark 2.6]{Weiss1989SICON},  \dref{201994746FHAD},  \dref{201994731FHAD} and the admissibility of $B_1$ for $e^{A_1 t}$  to get
\begin{equation}\label{wxh20189241014FHAD}
\left\|\int_{0}^{t}e^{A_1(t-s)}B_1v_{\omega}(s)ds\right\|_{X_1}\leq
M\|v_{\omega}\|_{ L^2([0, \infty);U_1)},\ \ \forall\ t> 0,
\end{equation}
where  $M>0$ is a  constant   independent of $t$.
 On the other hand, the solution of the $x_1$-subsystem is
 \begin{equation}\label{201994743FHAD}
 x_1(t)=e^{A_1t}x_1(0)+\int_0^te^{A_1(t-s)}B_1C_2x_2(s)ds\in X_1.
\end{equation}
Combining  \dref{201994746FHAD}, \dref{201994731FHAD} and  \dref{wxh20189241014FHAD},
 for any  $0<\theta<1$,  we have
\begin{equation*}\label{wxh201892111274FHAD}
\begin{array}{l}
  \left\| \displaystyle\int_{0}^{t}  e^{ A_1 (t-s)} B_1C_2x_2(s)ds \right\|_{ X_1}
\leq   \left\|\displaystyle\int_{0}^{\theta t}  e^{ A_1 (t-s)} B_1C_2x_2(s)ds\right\|_{ X_1}       +\left\|\displaystyle\int_{\theta t}^{t}   e^{ A_1 (t-s)} B_1C_2x_2(s)ds\right\|_{ X_1} \crr
\leq   \left\|\displaystyle e^{ A_1 (1-\theta)t}\int_{0}^{\theta t}e^{ A_1 (\theta t-s)}
 B_1C_2x_2(s)ds\right\|_{ X_1}  +e^{-\omega\theta t} \left\|\displaystyle \int_{\theta t}^{t}e^{ A_1 (t-s)} B_1v_{\omega}(s)ds\right\|_{ X_1} \crr
\leq   L_1e^{-\omega_1(1-\theta)t}M\| C_2x_2\|_{L^2([0, \infty);U_1)}+e^{-\omega\theta t}
M\|  v_{\omega}\|_{L^2([0, \infty);U_1) },
\end{array}
\end{equation*}
which, together with \dref{201994835FHAD}, \dref{201994743FHAD} and  \dref{201994746FHAD}, leads  to
 the exponential stability of $(x_1(t),x_2(t))$ in $X_1\times X_2$.
 The proof is  complete.
\end{proof}

Similarly  to   Lemma \ref{Lm201912031723}, we obtain immediately Lemma \ref{Lm202012907}.
 \begin{lemma}\label{Lm202012907}
     Let $X_1,X_2$ and $U_2$ be Hilbert spaces.
   Suppose that  $A_j$   generates  a
 $C_0$-semigroup  $e^{A_jt}$ on $X_j$,
 $B_2\in \mathcal{L}(U_2,[D(A_2^*)]')$
 is admissible for $e^{A_2t}$ and
 $C_1\in \mathcal{L}(  D(A_1) , U_2)$ is admissible for $e^{A_1t}$,
 $j=1,2$.
  Let
   \begin{equation} \label{202012912}
 \A_1=\begin{pmatrix}
\tilde{A}_1&0\\B_2 C_{1\Lambda} &\tilde{A}_2
\end{pmatrix},
\begin{array}{l}
 \disp
 D(\A_1 )=\left\{(x_1,x_2) ^{\top}\in D(A_1)\times  X_{2}  \  |\  \tilde{A}_2x_2+ B_2 C_{1\Lambda}x_1\in X_2 \right\}.
 \end{array}
 \end{equation}
Then, the operator $\A_1 $   generates a
  $C_0$-semigroup $e^{\A_1 t}$ on $X_1\times X_2$.
 Moreover, if we  suppose further that  $e^{A_it}$ is
   exponentially stable in $X_j$, $j=1,2$, then, $e^{\A_1t}$ is  exponentially stable in $X_1\times X_2$.
\end{lemma}
\begin{proof}
 The proof is almost the same as Lemma \ref{Lm201912031723} and we omit  the details.
\end{proof}

\section{ Sylvester  equations }\label{Se.4}

In view of \dref{20191024842}, we need extend the Sylvester    equation to the
infinite-dimensional cases.
For this  purpose, we first give the  definition of the solution of    Sylvester equation.

 \begin{definition}\label{Defin20191023}
  Let $X_1$, $X_2$ and  $U_1$ be Hilbert spaces  and $A_j: D(A_j)\subset X_j\to X_j$
      be a densely defined
      operator with $\rho(A_j)\neq \emptyset$, $j=1,2$.
        Suppose that  $B_1\in \mathcal{L} ( U_1 , [D(A_1^*)]')$ and  $C_2\in \mathcal{L} (D(A_2),U_1)$.
 We say  that the     operator  $ S  $  is a  solution of the   Sylvester equation
  \begin{equation} \label{201910231059}
 \left.\begin{array}{l}
\disp       {A}_1  S - S   {A}_2 =   B_1C_{2}\ \ \mbox{on}\ \ D(A_2),
\end{array}\right.
\end{equation}
 if  $ S   \in \mathcal{L}(X_2,X_1)$ and  the following equality holds
   \begin{equation} \label{201910231103}
   \tilde{A}_1  S  x_2  -  S    {A}_2x_2  =   B_1C_{2} x_2 ,\ \ \forall\; x_2\in D(A_2),
\end{equation}
where $\tilde{A}_1$ is an extension of $A_1$  given by  \dref{20191121602}.
  \end{definition}

  \begin{lemma}\label{Lm2020629721}
   Suppose that $X_j$,   $U_j$ and $Y_j$ are  Hilbert spaces,  $A_j: D(A_j)\subset X_j\to X_j$
      is a densely defined
      operator with $\rho(A_j)\neq \emptyset$,
  $B_j\in \mathcal{L} ( U_j , [D(A_j^*)]')$ and  $C_j\in \mathcal{L} (D(A_j),Y_j)$,
   $j=1,2$.
  Let
    \begin{equation} \label{2020629801}
X_{jB_j}=D(A_j)+(\lambda_j-\tilde{A}_j)^{-1}B_jU_j,\ \ \lambda_j\in \rho(A_j), \ \ j=1,2.
\end{equation}
 Then,
  $X_{jB_j}$    is independent of $\lambda_j$ and can be characterized as
     \begin{equation} \label{2020629814}
 X_{jB_j}=\left\{ x_j\in X_j\ | \  \tilde{A}_jx_j+B_ju_j \in X_j,\  u_j\in U_j \right\},\ \   j=1,2 .
\end{equation}
  Suppose further that $Y_2=U_1$ and
   $S  \in \mathcal{L}(X_2, X_1)$  is  a solution of the Sylvester equation \dref{201910231059} in the sense of Definition \ref{Defin20191023}.
  Then,
  the following assertions hold  true:

 (i). If    $(A_1,B_1,C_1)$ is a regular linear system,   then,
      $C_{1\Lambda}S\in \mathcal{L}(D(A_2), Y_1) $;

(ii). If    $(A_2,B_2,C_2)$ is a regular linear system,    then,
 there exists an extension of $S$, still denoted by $S$,  such that
$ SB_2\in \mathcal{L}(U_2, [D(A_1^*)]') $ and
   \begin{equation} \label{2020629754}
   \tilde{A}_1  S  x_2  -
  S   \tilde{A}_2x_2  =   B_1C_{2\Lambda} x_2, \ \ \forall\;  x_2\in X_{2B_2}.
\end{equation}
  \end{lemma}
  \begin{proof}
    The definition of the  space $ X_{jB_j}$ and  its characterization \dref{2020629814} can be obtained
     by \cite[Section 2.2]{Salamon1987} and \cite[Remark 7.3]{Weiss1994MCSS} directly.

     The proof of (i).
 Since
 $S   $  solves the  Sylvester equation \dref{201910231059},
for any $x_2\in D(A_2)$,   we have
 $ \alpha Sx_2-\tilde{A}_1  S  x_2 +S    {A}_2x_2 = \alpha Sx_2-B_1C_{2} x_2$ with $\alpha \in \rho(A_1)$. That is
  \begin{equation} \label{2020822958}
  Sx_2= (\alpha-\tilde{A}_1)^{-1}S(\alpha- A_2)x_2-(\alpha-\tilde{A}_1)^{-1}B_1C_{2} x_2,\ \ \forall\ x_2\in D(A_2).
  \end{equation}
  Since   $(A_1,B_1,C_1)$ is a regular linear system and  $S  \in \mathcal{L}(X_2, X_1)$,  \dref{2020822958} implies that  $C_{1\Lambda}S\in \mathcal{L}(D(A_2), Y_1)$.

  The proof of (ii).    In terms of the solution  $S  \in \mathcal{L}(X_2, X_1)$  of  \dref{201910231059}, we define the operator $\tilde{S}$ by
   \begin{equation} \label{2020629855}
   \tilde{S}=   B_1C_{2\Lambda}(\beta-\tilde{A}_2)^{-1}+(\beta -\tilde{A}_1) S(\beta-\tilde{A}_2)^{-1},\ \ \beta\in \rho(A_2).
  \end{equation}
For  any $x_2\in X_2$,  since  $(\beta-\tilde{A}_2)^{-1}x_2\in D(A_2)$,
 it follows from  \dref{201910231103} that
     \begin{equation} \label{2020629858}
     \begin{array}{ll}
 \disp   \tilde{S}x_2&\disp =   B_1C_{2\Lambda}(\beta-\tilde{A}_2)^{-1}x_2-\tilde{A}_1S(\beta-\tilde{A}_2)^{-1}x_2+\beta S(\beta-\tilde{A}_2)^{-1}x_2\crr
 &\disp =   -S  \tilde{A}_2(\beta-\tilde{A}_2)^{-1}x_2+S\beta (\beta-\tilde{A}_2)^{-1}x_2 \crr
 &\disp =S(\beta-\tilde{A}_2) (\beta-\tilde{A}_2)^{-1}x_2=Sx_2,
 \end{array}
  \end{equation}
  which implies that $\tilde{S} $ is an extension of $S$.   On the other hand,     by the  regularity of      $(A_2,B_2,C_2)$ and the definition  \dref{2020629855}, we can conclude that $(\beta-\tilde{A}_2)^{-1}B_2 \in \mathcal{L} (U_2,D(C_{2\Lambda}))  $  and
   \begin{equation} \label{2020629924}
     \begin{array}{l}
 \disp   \tilde{S} B_2  \disp =   B_1C_{2\Lambda}(\beta-\tilde{A}_2)^{-1}B_2 +\beta S(\beta-\tilde{A}_2)^{-1}B_2 -\tilde{A}_1S(\beta-\tilde{A}_2)^{-1}B_2,
 \end{array}
  \end{equation}
  which implies that $\tilde{S}B_2\in \mathcal{L}(U_2, [D(A_1^*)]')$.  Moreover,
  for any $u_2\in U_2$,
 it follows from  \dref{2020629858} and \dref{2020629924} that
     \begin{equation} \label{20205281647}
         \begin{array}{l}
\disp \tilde{A}_1  \tilde{S} [(\beta-\tilde{A}_2)^{-1}B_2u_2]  -
 B_1C_{2\Lambda}[(\beta-\tilde{A}_2)^{-1}B_2u_2]= \tilde{S}\beta(\beta-\tilde{A}_2)^{-1}B_2u_2-\tilde{S}B_2u_2\crr
=\tilde{S}\beta(\beta-\tilde{A}_2)^{-1}B_2u_2-\tilde{S} (\beta-\tilde{A}_2) (\beta-\tilde{A}_2)^{-1}  B_2u_2=\tilde{S}\tilde{A}_2[(\beta-\tilde{A}_2)^{-1}B_2u_2].
\end{array}
  \end{equation}
   Due to  the arbitrariness of $u_2$, \dref{20205281647} implies that
  $\tilde{S}  $  solves  the Sylvester equation \dref{201910231059} on $(\beta-\tilde{A}_2)^{-1}B_2U_2$.
  Since  $\tilde{S}|_{X_2}=S$, \dref{201910231103}  and \dref{2020629801},
  we can obtain  \dref{2020629754}  easily   with replacement of  $S$ by  $\tilde{S}$. The proof is complete.
  \end{proof}


\begin{lemma}\label{20204292101}
  Let $X_1$, $X_2$ and   $U_1$  be   Hilbert spaces and let  $A_1 : D(A_1)\subset  X_1\to X_1$ be a
  densely defined    operator with $\rho(A_1) \neq \emptyset$.  Suppose that $A_2\in\mathcal{L}(X_2)$,   $C_2\in \mathcal{L} (X_2,U_1) $,
  $B_1\in \mathcal{L} ( U_1,[D(A_1^*)]') $  and
       \begin{equation} \label{20204292103}
 \sigma(A_1)\cap\sigma(A_2)=\emptyset.
\end{equation}
    Then, the Sylvester equation \dref{201910231059}   admits a  solution
      $S  \in \mathcal{L}(X_2, X_1)$ in the sense of Definition \ref{Defin20191023}.
 \end{lemma}
\begin{proof}
   Since
 $\tilde{A}_1 \in\mathcal{L}(X_1,[D(A_1^*)]')$ and
   $B_1C_{2}\in\mathcal{L}(X_2,[D(A_1^*)]')$, it follows from
   \cite[Lemma 22]{Sylvester1991} and \dref{20204292103}  that
  the following Sylvester  equation
   \begin{equation} \label{202058805}
  \tilde{A}_1S-S {A}_2 =B_1C_{2 }
\end{equation}
   admits a unique solution $S  \in \mathcal{L}(   X_2,[D(A_1^*)]')$ in the sense that
   $S(X_2)=S(D(A_2))\subset D(\tilde{A}_1)=X_1$  and
   \begin{equation} \label{202058811}
   \tilde{A}_1Sx_2-S {A}_2 x_2=B_1C_{2 }x_2,\  \ \forall\ x_2\in X_2.
  \end{equation}
 By  a simple computation,  we have
  \begin{equation} \label{2020527919}
S=(\alpha  -\tilde{A}_1)^{-1}S(\alpha-A_2)-(\alpha  -\tilde{A}_1)^{-1}B_1C_2,\ \ \alpha\in \rho(A_1),
  \end{equation}
  which, together with the fact
 $(\alpha  -\tilde{A}_1)^{-1}\in \mathcal{L}([D(A_1^*)]',X_1)$,
   implies that $S  \in \mathcal{L}(   X_2,X_1)$.
   This shows that
  $S   $ is a  solution of
  equation \dref{201910231059}    in the sense of Definition \ref{Defin20191023}.
  The proof is complete.
  \end{proof}

\begin{lemma}\label{lm2020528744}
  Let $X_1$, $X_2$ and  $U_1$  be  Hilbert spaces and let
    $A_2 : D(A_2)\subset  X_2\to X_2$ be a
  densely defined   operator with $\rho(A_2) \neq \emptyset$.  Suppose that
$A_1\in\mathcal{L}(X_1)$,  $B_1\in \mathcal{L} (U_1 , X_1)$,  $C_2\in \mathcal{L} (D(A_2),U_1 ) $
  and   \dref{20204292103} holds.
 Then, the Sylvester equation \dref{201910231059}   admits a  solution
      $S  \in \mathcal{L}( X_2, X_1)$ in the sense of Definition \ref{Defin20191023}.
\end{lemma}
\begin{proof}
 Since $A_2$ is densely defined and $\rho(A_2) \neq \emptyset$, $A_2$ is closed.
It follows from  \cite[Proposition 2.8.1, p.53]{TucsnakWeiss2009book} that  $A_2^{**}=A_2$  and thus  $\widetilde{A_2^*}\in \mathcal{L}(X_2,[D(A_2)]')$, where
 $\widetilde{A_2^*}$ is an extension of $A_2^*$  given by  \dref{20191121602}.
  Moreover, it follows from    \cite[Theorem 5.12, p.99]{Weidmann1980book} and   \dref{20204292103} that $\sigma(A^*_1)\cap\sigma(A^*_2)= \emptyset$.
 By Lemma \ref{20204292101},    there exists a
solution  $\Pi\in \mathcal{L}(X_1,X_2)$  to  the Sylvester equation $\Pi A_1^*-A_2^*\Pi=C_2^*B_1^*$.
 In particular, for any $x_2\in D(A_2)$, $ x_1\in X_1$,
   \begin{equation} \label{2020527957}
 \left.\begin{array}{l}
\disp        \langle  \widetilde{A_2^*}  \Pi  x_1,x_2\rangle_{[D(A_2)]',D(A_2)} -
\langle \Pi  {A}_1^*x_1,x_2\rangle_{X_2} =   -\langle C_{2 }^*B_1^* x_1, x_2\rangle_{[D(A_2)]',D(A_2)}.
\end{array}\right.
\end{equation}
That is
   \begin{equation} \label{20205271057}
 \left.\begin{array}{l}
\disp        \langle    x_1,\Pi ^*A_2x_2\rangle_{X_1} -
\langle x_1, {A}_1 \Pi ^*x_2\rangle_{X_1} =   -\langle x_1,B_1 C_{2 }x_2\rangle_{X_1}.
\end{array}\right.
\end{equation}
Since  $x_1$ is arbitrary, the equality
$\Pi ^*A_2x_2  -
  {A}_1 \Pi ^*x_2  =   - B_1 C_{2 }x_2$  holds in $X_1$ for any
 $x_2\in D(A_2)$.
 Therefore, $S=\Pi^* \in \mathcal{L}(X_2,X_1)$ is a solution of
  equation \dref{201910231059}.
  The proof is complete.
\end{proof}

\section{Actuator dynamics compensation}\label{Se.5}
This section is devoted to the extension of the  results in Section \ref{Se.2} from finite-dimensional
systems to  the  infinite-dimensional ones.

\begin{assumption}\label{A1}
Let $X_1$, $X_2$, $U_1$ and $U_2$ be   Hilbert spaces.  The operator $A_j$ generates  a $C_0$-semigroup $e^{A_jt}$ on $X_j$,
  $B_j\in \mathcal{L}(U_j,[D(A_j^*)]')$ is admissible for $e^{A_jt}$ and  $C_2
  \in \mathcal{L}(D(A_2),U_1) $ is admissible for
 $e^{A_2t}$, $j=1,2$. In addition,
   $\sigma(A_1) \cap \sigma(A_2)=\emptyset$
 and the semigroup $e^{A_2 t} $ is exponentially stable in $X_2$.
\end{assumption}

Since  the  stabilization  and compensation   of the actuator dynamics are two different  issues, we assume  additionally that
   the semigroup $e^{A_2 t} $ is exponentially stable, which is just  for avoidance of  the confusion.
Indeed, one just needs to stabilize the system  before the actuator dynamics compensation  when $e^{A_2 t} $ is not exponentially stable.
Since $e^{A_2 t} $ is exponentially stable already, the full state feedback  of system \dref{201911101138} can be designed, inspired by \dref{201911281134}, as
 \begin{equation} \label{201911281134infinite}
u(t)= K_{1\Lambda}x_1(t)+K_{1\Lambda}Sx_2(t),
\end{equation}
where
 the operator $ S  \in \mathcal{L}( X_2, X_1)$   is a solution  of
  the Sylvester equation
 \begin{equation} \label{201911301901}
    A_1 S-S   A_2  =   B_1C_2,
\end{equation}
and $K_1\in \mathcal{L}(D(A_1), U_2)$ is selected such that
 $A_1+ {S}B _2K_{1\Lambda}$  generates an exponentially stable  $C_0$-semigroup   on $X_1$.
 Under   controller \dref{201911281134infinite}, we obtain
 the
closed-loop
system:
 \begin{equation} \label{201911292049}
 \left\{\begin{array}{l}
\disp \dot{x}_1(t) = A_1 x_1(t)+B_1C_{2\Lambda}x_2(t),\crr
\disp  \dot{x}_2(t) =   (A_2  +B_2 K_{1\Lambda}S)x_2(t) +B_2 K_{1\Lambda}{x}_1(t).
\end{array}\right.
\end{equation}
 Define
\begin{equation} \label{201912021831}
\left\{\begin{array}{l}
\disp \mathscr{A} = \begin{pmatrix}
\tilde{A}_1& B_1C_{2\Lambda} \\
B_2K_{1\Lambda} &\tilde{A}_2 +B_2K_{1\Lambda}S
\end{pmatrix}, \crr
\disp  D(\mathscr{A})= \left\{  \begin{pmatrix}
                                 x_1 \\
                                 x_2
                               \end{pmatrix} \in X_{ 1}\times X_{ 2}\ \Big{|}\
\begin{array}{l} \tilde{A}_1x_1+B_1C_{2\Lambda}x_2\in X_1  \\
 \disp  B_2K_{1\Lambda}x_1+ (\tilde{A}_2+B_2K_{1\Lambda}S)x_2\in X_2\\\end{array}\right\}.
\end{array}\right.
\end{equation}
  Then, the
closed-loop
system \dref{201911292049} can be written as
 \begin{equation} \label{201912021832}
 \frac{d}{dt} (x_1(t),x_2(t))^{\top}=\mathscr{A}   (x_1(t),x_2(t))^{\top}.
   \end{equation}
In view of \dref{201910201037A2}, we define the operator
  \begin{equation} \label{20206291024}
\mathscr{A}_{S}= \begin{pmatrix}
\tilde{A}_1+ {S}B _2K_{1\Lambda} & 0\\
 B_2K_{1\Lambda }&  \tilde{A}_2
\end{pmatrix}
 \end{equation}
 with
 \begin{equation} \label{20206291025}
 D(\mathscr{A}_{S})  =\left\{(x_1,x_2)^{\top}\in  X_{1}  \times X_{2} \ |\
  ( \tilde{A}_1 + {S}B _2K_{1\Lambda}) x_1\in X_1,
 \tilde{A}_2x_2+  B_2K_{1\Lambda } x_1\in X_2 \right\}.
 \end{equation}

\begin{theorem}\label{Th201912021844}
In addition to   Assumption \ref{A1}, suppose that   $A_1\in \mathcal{L} (X_1) $ and  $(A_2,B_2,C_2)$ is a regular linear system.  Then,
the Sylvester equation
 \dref{201911301901} admits a  solution
$ S  \in \mathcal{L}( X_2, X_1)$ in the sense of Definition \ref{Defin20191023} such that
$SB_2\in \mathcal{L}(U_2,X_1)$.
If we suppose further that
  there is a $K_1\in \mathcal{L}(X_1, U_2)$   such that
  $A_1+{S}B _2K_{1 }$  generates an exponentially stable  $C_0$-semigroup   $e^{(A_1+ {S}B _2K_{1 })t}$  on $X_1$,
  then, the operator $\mathscr{A}$ defined by \dref{201912021831}  generates an exponentially stable   $C_0$-semigroup $e^{\mathscr{A} t}$
  on $X_1\times X_2$.
  \end{theorem}
\begin{proof}
Since  $A_1\in \mathcal{L} (X_1) $, it follows that $B_1\in \mathcal{L}(U_1,X_1)$ and
$K_1=K_{1\Lambda}\in \mathcal{L}(X_1, U_2)$. By Lemmas  \ref{Lm2020629721} and \ref{lm2020528744},
the Sylvester equation
 \dref{201911301901} admits a solution
$ S  \in \mathcal{L}( X_2, X_1)$
  such that
$SB_2\in \mathcal{L}(U_2,X_1)$ and
  \begin{equation} \label{20206291031}
   {A}_1  S  x_2  -
  S   \tilde{A}_2x_2  =   B_1C_{2\Lambda} x_2 , \ \ \forall\; x_2\in X_{2B_2},
\end{equation}
where $X_{2B_2}$ is defined by \dref{2020629801}   or equivalently by \dref{2020629814}.
  Moreover, $SB_2K_1\in \mathcal{L}( X_1)$ and thus
 \begin{equation} \label{20205271531}
 D(\mathscr{A}_{S})  =\left\{(x_1,x_2)^{\top}\in  X_{1}  \times X_{2} \ |\
 \tilde{A}_2x_2+  B_2K_{1 } x_1\in X_2 \right\}.
 \end{equation}
We claim that $\mathscr{A}\sim_\mathbb{S}\mathscr{A}_{S} $, i.e.,
  \begin{equation} \label{20206261825}
 \left.\begin{array}{l}
\disp \mathbb{S}\mathscr{A} \mathbb{S} ^{-1}=\mathscr{A}_{S} \ \ \mbox{and}\ \  D(\mathscr{A}_{S}) = \mathbb{S}D(\mathscr{A}),
\end{array}\right.
\end{equation}
where the   transformation $\mathbb{S} $ is given  by
    \begin{equation} \label{20205272026}
   \left.\begin{array}{l}
\disp    \mathbb{S}  \left( {x}_1 , {x}_2 \right)^{\top}= \left( {x}_1+S  {x}_2,\ {x}_2 \right)^{\top},\ \ \forall\ ( {x}_1, {x}_2)^\top\in  X_1\times X_2.
 \end{array}\right.
\end{equation}
Obviously, $\mathbb{S}\in\mathcal{L}(X_1\times X_2)$  is invertible and its inverse is given by
\begin{equation} \label{20205272027}
   \left.\begin{array}{l}
 \disp    \mathbb{S}^{-1} \left( {x}_1 , {x}_2 \right)^{\top}= \left(   {x}_1- S  {x}_2,
 {x}_2\right)^{\top},\ \ \forall \left( {x}_1 , {x}_2 \right)^{\top}\in  X_1\times X_2.
 \end{array}\right.
\end{equation}
For any $(x_1,x_2)^\top\in D(\mathscr{A}_{S})$,  we have
$\tilde{A}_2x_2+  B_2K_{1 } x_1 \in X_2$ and $K_1x_1\in U_2$.
It follows from   \dref{2020629814} and the regularity of
 $(A_2,B_2,C_2)$  that
 $x_2\in X_{2B_2} \subset D(C_{2\Lambda})$. As a result, $B_1C_{2\Lambda}x_2\in X_1$ and thus
\begin{equation} \label{20206261840}
 A_1(x_1-Sx_2)+B_1C_{2\Lambda}x_2 \in X_1.
\end{equation}
Since
\begin{equation} \label{20206261842}
B_2K_1(x_1-Sx_2)+\tilde{A}_2x_2+B_2K_1Sx_2= \tilde{A}_2x_2+B_2K_1 x_1
 \in X_2,
\end{equation}
we combine  \dref{20205272027}, \dref{20206261840},  \dref{20206261842} and \dref{201912021831} to get  $\mathbb{S}^{-1}(x_1,x_2)^{\top}\in D(\mathscr{A})$.  Consequently, $D(\mathscr{A}_{S})\subset \mathbb{S}D(\mathscr{A})$
due to the arbitrariness of $(x_1,x_2)^\top\in D(\mathscr{A}_{S})$.
On the other hand, for any $(x_1,x_2)^\top\in D(\mathscr{A} )$,
By \dref{20205271531},
\dref{20205272026} and $\tilde{A}_2x_2+B_2K_1(x_1+Sx_2)\in X_2$,
we get
$\mathbb{S}(x_1,x_2)^\top\in D(\mathscr{A}_{S})$ and thus
$\mathbb{S}D(\mathscr{A})\subset D(\mathscr{A}_{S}) $.
We have therefore obtained that  $D(\mathscr{A}_{S})= \mathbb{S}D(\mathscr{A})$.

  For any  $(x_1,x_2)^\top\in D(\mathscr{A}_{S})$, it follows from \dref{20205271531},  \dref{2020629814}  and the regularity of
 $(A_2,B_2,C_2)$  that $x_2\in X_{2B_2}$. By virtue of \dref{20206291031},
 a straightforward computation shows that  $ \mathbb{S}\mathscr{A} \mathbb{S} ^{-1}(x_1,x_2)^{\top}=\mathscr{A}_{S}(x_1,x_2)^{\top}$ for any $(x_1,x_2)^{\top}\in D(\mathscr{A}_{S})$. Consequently,
  $\mathscr{A}$  and $\mathscr{A}_S$  are similar each other.

Since   the   $C_0$-semigroups   $e^{(A_1+ {S}B _2K_1)t}$  on $X_1$ and
 $e^{A_2t}$  on $X_2$  are  exponentially stable, $K_1\in \mathcal{L}(X_1, U_2)$ and
  $B_2$ is admissible for $e^{A_2t}$,
  it follows from Lemma \ref{Lm202012907} that the operator $\mathscr{A}_{S}$ generates an exponentially stable
   $C_0$-semigroup  $e^{\mathscr{A}_S t}$ on $X_1\times X_2$.
   By  the similarity of $\mathscr{A}_S$ and $\mathscr{A}$,
    the operator $\mathscr{A} $ generates an exponentially stable
   $C_0$-semigroup  $e^{\mathscr{A} t}$ on $X_1\times X_2$ as well.
  This completes the proof of the theorem.
 \end{proof}

When $X_1$ is finite-dimensional, we can characterize the existence of the feedback gain    $K_1$
through the system \dref{201911101138}  itself.

\begin{corollary}\label{Co201912036820}
In addition to Assumption \ref{A1}, suppose  that   $X_1$ is
finite-dimensional,
 $(A_2,B_2,C_2)$ is a regular linear system and
   system \dref{201911101138} is  approximately controllable.
  Then,
  there exist  $S\in \mathcal{L}(X_2,X_1)$ and
  $K_1\in \mathcal{L}(X_1, U_2)$ such that
    the operator $\mathscr{A}$  defined by \dref{201912021831} generates an exponentially stable   $C_0$-semigroup $e^{\mathscr{A} t}$
  on $X_1\times X_2$.
 \end{corollary}

\begin{proof}
 By Lemmas \ref{lm2020528744} and \ref{Lm2020629721},
the Sylvester equation
 \dref{201911301901} admits a   solution
$ S  \in \mathcal{L}( X_2, X_1)$   such that
$SB_2\in \mathcal{L}(U_2,X_1)$ and \dref{2020629754} holds.
 Define
\begin{equation} \label{20204291140}
 \A_{S} =\begin{pmatrix}
A_1&0\\0&\tilde{A}_2
\end{pmatrix}\ \ \mbox{with}\ \ D(\A_{S} )=X_1\times D(A_2),
\
 \mathcal{B}_2 =\begin{pmatrix}
 0 \\
 B_2
\end{pmatrix},\
\mathcal{B}_{S}= \begin{pmatrix}
SB_2   \\
 B_2
\end{pmatrix}.
 \end{equation}
 A simple computation shows that $\A\sim_\mathbb{S}\A_{S}$, i.e.,
  $\mathbb{S}\A\mathbb{S}^{-1}=\A_S$ and $D(\A_S)=\mathbb{S}D(\A)$,
  where the operator $\A$ is given by \dref{2020428921}
    and
$\mathbb{S}$ is given  by \dref{20205272026}.  Moreover,
 $\mathcal{B}_{S}=\mathbb{S} \mathcal{B}_2 $ satisfies
  \begin{equation*} \label{2020581133}
 \left.\begin{array}{l}
\disp \left\langle \mathcal{B}_{S}u,
\begin{pmatrix}
  x_1 \\
  x_2
\end{pmatrix}
\right\rangle_{ [D(\A_{S}^*)]', D(\A_{S}^* )}=
\left\langle  \mathcal{B}_2  u,\mathbb{S}^*\begin{pmatrix}
  x_1 \\
  x_2
\end{pmatrix}
\right\rangle_{[D(\A ^*)]',D(\A^*)},
   \forall\    u\in U_2, \begin{pmatrix}
  x_1 \\
  x_2
\end{pmatrix}\in D(\A_{S}^*  ).
\end{array}\right.
 \end{equation*}
By Lemma \ref{th201911262002}  and the approximate controllability of system
  $(\A, \mathcal{B}_2)$, system
$(\A_{S}, \mathcal{B}_{S})$ is approximately controllable.
Thanks to the block-diagonal structure of $\A_{S}$, it follows from Lemma \ref{Lm20205291214} of Appendix that  the finite-dimensional system     $(A_1,  SB_2)$ is   controllable.
 By  the pole assignment theorem, there exists a
 $K_1\in \mathcal{L}(X_1, U_2)$   to stabilize system $(A_1,  SB_2)$.
  By Theorem \ref{Th201912021844},
  $\mathscr{A}$   generates an exponentially stable   $C_0$-semigroup $e^{\mathscr{A} t}$
  on $X_1\times X_2$.
\end{proof}
\begin{theorem}\label{Th20205271541}
In addition to Assumption \ref{A1}, suppose that   $A_2\in \mathcal{L} (X_2) $, $K_1\in \mathcal{L}(D(A_1), U_2)$  and $(A_1,B_1, K_1 )$ is a regular linear system.   Then,
the Sylvester equation
 \dref{201911301901} admits a   solution
$ S  \in \mathcal{L}( X_2, X_1)$ in the sense of Definition \ref{Defin20191023} and   $SB_2\in \mathcal{L}(U_2,X_1)$.
If we suppose  further that
 $K_1$
 stabilizes  system
  $(A_1, {S}B _2)$   in the sense of  \cite[Definition 3.1]{Weiss1997TAC},
   then,   the operator $\mathscr{A}$ defined by \dref{201912021831}  generates an exponentially stable   $C_0$-semigroup $e^{\mathscr{A} t}$
  on $X_1\times X_2$.
\end{theorem}
\begin{proof}
Since  $(A_1,B_1,K_1)$ is   a regular linear system,
by   Lemmas  \ref{Lm2020629721} and \ref{20204292101},  the Sylvester equation \dref{201910231059}
admits a   solution
  $ S   \in \mathcal{L}(X_2,X_1)$      such that $K_{1\Lambda}S\in \mathcal{L}(X_2,U_2)$ and
   \begin{equation} \label{20206291056}
   \tilde{A}_1  S  x_2  -
  S  {A}_2x_2  =   B_1C_{2 } x_2 ,\ \ \forall\ \ x_2\in X_2.
\end{equation}
 Since $A_2\in \mathcal{L} (X_2) $ and $B_2\in \mathcal{L} (U_2, X_2) $, we have
  $SB_2\in \mathcal{L}(U_2,X_1)$ and $D(\mathscr{A}_{S})$ in \dref{20206291025} becomes
   \begin{equation} \label{2020627917}
 D(\mathscr{A}_{S})  =\left\{(x_1,x_2)^{\top}\in  X_{1}  \times X_{2} \ |\
\tilde{A}_1x_1+ {S}B _2K_{1\Lambda} x_1\in X_1, B_2K_{1\Lambda}x_1\in X_2 \right\}.
 \end{equation}
Similarly to \dref{20206261825}, we claim that $\mathscr{A}\sim_\mathbb{S}\mathscr{A}_{S} $, i.e.,
$\mathbb{S}\mathscr{A} \mathbb{S} ^{-1}=\mathscr{A}_{S}$  and  $ D(\mathscr{A}_{S}) = \mathbb{S}D(\mathscr{A})$,
where
the  transformation $\mathbb{S}\in\mathcal{L}(X_1\times X_2)$ is given  by
    \dref{20205272026}.
Actually, for any $(x_1,x_2)^\top\in D(\mathscr{A}_{S})$,
it  follows from
 \dref{20206291056} that
\begin{equation} \label{20206291638}
 \tilde{A}_1(x_1-Sx_2)+B_1C_{2 }x_2 =
 \tilde{A}_1 x_1-SA_2x_2=(\tilde{A}_1 x_1+SB_2K_{1\Lambda}x_1)- S(B_2K_{1\Lambda}x_1+A_2x_2),
 \end{equation}
which, together with   $S\in \mathcal{L}(X_2,X_1)$,    $A_2\in \mathcal{L}( X_2)$ and \dref{2020627917},
  leads to $ \tilde{A}_1(x_1-Sx_2)+B_1C_{2 }x_2\in X_1$.
   Since  $K_{1\Lambda}S\in \mathcal{L}(X_2,U_2)$, $B_2\in\mathcal{L}(U_2,X_2)$ and  $B_2K_{1\Lambda}x_1\in X_2$, we have
 \begin{equation} \label{2020627922}
B_2K_{1\Lambda}(x_1-Sx_2)+ {A}_2x_2+B_2K_{1\Lambda}Sx_2=
{A}_2x_2+B_2K_{1\Lambda} x_1
 \in X_2.
\end{equation}
  In view of \dref{201912021831} and   $\mathbb{S}^{-1}(x_1,x_2)^{\top}=(x_1-Sx_2,x_2)^\top$,
  we can conclude  that $\mathbb{S}^{-1}(x_1,x_2)^{\top}\in D(\mathscr{A})$.  Consequently, $D(\mathscr{A}_{S})\subset \mathbb{S}D(\mathscr{A})$
due to the arbitrariness of $(x_1,x_2)^\top\in D(\mathscr{A}_{S})$.

On the other hand, for any $(x_1,x_2)^\top\in D(\mathscr{A} )$,
 since $K_{1\Lambda}S\in \mathcal{L}(X_2,U_2)$, $S\in \mathcal{L}(X_2,X_1)$, $B_2\in \mathcal{L}(U_2,X_2)$ and $A_2\in \mathcal{L}( X_2)$,  \dref{201912021831} yields
$B_2K_{1\Lambda }x_1\in X_2$ and $B_2K_{1\Lambda}S x_2\in X_2$. As a result,
\begin{equation} \label{20206271008}
B _2K_{1\Lambda} (x_1+Sx_2) \in X_2\ \ \mbox{and}\ \ {S}B _2K_{1\Lambda} (x_1+Sx_2)\in X_1.
\end{equation}
 It follows from
 \dref{20206291056} and \dref{201912021831} that
 \begin{equation} \label{20206271009}
\tilde{A}_1(x_1+Sx_2)=(\tilde{A}_1 x_1+B_1C_2x_2)+SA_2x_2 \in X_1.
\end{equation}
 Combining \dref{20206271008},  \dref{20206271009} and \dref{2020627917},
 we  can conclude that
$\mathbb{S}(x_1,x_2) ^{\top} =(x_1+Sx_2,x_2) ^{\top} \in D(\mathscr{A}_{S})$ and   thus
$\mathbb{S}D(\mathscr{A})\subset D(\mathscr{A}_{S}) $.
 Consequently, we obtain $D(\mathscr{A}_{S})= \mathbb{S}D(\mathscr{A})$. By a straightforward computation, we also have $ \mathbb{S}\mathscr{A} \mathbb{S} ^{-1}(x_1,x_2)^\top=\mathscr{A}_{S}(x_1,x_2)^\top$
 for any $(x_1,x_2)^\top\in D(\mathscr{A}_{S})$.
  This shows that
  $\mathscr{A}$  and $\mathscr{A}_S$  are similar each other.

Since   $K_1\in \mathcal{L}(D(A_1), U_2)$
   stabilizes system
  $(A_1, {S}B _2)$ exponentially,   the operator $\tilde{A}_1+ {S}B _2K_{1\Lambda}$
generates an exponentially stable
   $C_0$-semigroup   $e^{(\tilde{A}_1+ {S}B _2K_{1\Lambda})t}$  on $X_1$ and $K_1$
is admissible for $e^{(\tilde{A}_1+ {S}B _2K_{1\Lambda})t}$.
Since
 $e^{A_2t}$ is exponentially stable  and
  $B_2\in \mathcal{L}(U_2,X_2)$,
  it follows from Lemma \ref{Lm202012907} that the operator $\mathscr{A}_{S}$ generates an exponentially stable
   $C_0$-semigroup  $e^{\mathscr{A}_S t}$ on $X_1\times X_2$.
   By  the similarity of $\mathscr{A}_S$ and $\mathscr{A}$,
    the operator $\mathscr{A} $ generates an exponentially stable
   $C_0$-semigroup  $e^{\mathscr{A} t}$ on $X_1\times X_2$ as well.
          This completes the proof of the theorem.
\end{proof}

At the end of  this section, let us summarize
the   scheme  of the actuator dynamics compensation.
Given an  actuator dynamics compensation problem,
      we can design a  compensator  through five  steps:
      \begin{itemize}

      \item Formulate the control plant as the abstract form  \dref{201911101138};

      \item  Find $K_2$  to stabilize system   $ (A_2, B_2)$;

      \item  Solve the   Sylvester equation  $A_1S-S(A_2+B_2K_2)=B_1C_2 $;

  \item Find $K_1$  to stabilize   system   $(A_1, SB_2)$;

  \item  With $K_1$, $K_2$ and $S$ at hand, the controller is designed as
 \begin{equation} \label{201912031020}
u(t)=K_{2\Lambda}x_2(t)+K_{1\Lambda}x_1(t)+K_{1\Lambda}Sx_2(t).
\end{equation}

\end{itemize}

 We need   an  explicit expression of the solution of the  Sylvester operator equation to get   the controller.
Generally speaking, it is not easy to solve the Sylvester   equation.
However, under some reasonable additional assumptions, we still can obtain the solution  analytically or numerically even for  the cascade system involving  a multi-dimensional PDE.
 (see, e.g., \cite{Lassi2014TAC}  and \cite{Lassi2014SIAM}).
 In particular, when the  system \dref{201911101138} consists of an ODE and a one-dimensional PDE,  the problem becomes quite easy.
 Indeed, if  $X_1$ is    $n$-dimensional, we can suppose that
 \begin{equation} \label{2020529729}
\disp   Sx_2=
\begin{pmatrix}
\langle x_2,  \Phi_1\rangle_{X_2}\\
\langle x_2,  \Phi_2\rangle_{X_2}\\
\vdots\\
\langle x_2,  \Phi_n\rangle_{X_2}\\
\end{pmatrix}  :=\langle x_2,\Phi\rangle_{X_2}     ,\ \ \forall\ x_2\in X_2,
\end{equation}
 where
$\Phi=(\Phi_1,\Phi_2,\cdots,\Phi_n)^{\top} $  with
   $\Phi_j\in  X_2$, $j=1,2,\cdots,n$. Inserting \dref{2020529729} into the corresponding
  Sylvester equation, we will   arrive at a  vector-valued ODE  with respect to the variable $\Phi$.
   When  $X_2 $ is of  $m$-dimension,
  we can  suppose   that
  \begin{equation} \label{2020529739}
\disp   Sh=\sum_{j=1}^{m}\Psi_jh_j :=\langle  \Psi,h\rangle_{X_2}
,\ \ \forall\ h=(h_1,h_2,\cdots,h_m)^{\top}\in X_2,
\end{equation}
where $\Psi=(\Psi_1,\Psi_2,\cdots,\Psi_m)^{\top}$ with  $\Psi_j\in X_1$, $j=1,2,\cdots,m$.
Inserting \dref{2020529739} into the corresponding
  Sylvester equation will  lead to
  a  vector-valued  ODE    as well.
Thanks to the ODE theory and numerical analysis theory \cite{Thomasbook2}, in both cases, the solution $S$  can be obtained
 analytically  or numerically.   In this way,
we can stabilize the ODE with
 actuator dynamics,  dominated by the
   transport  equation  \cite{SmyshlyaevKrstic2008SCL},   wave equation  \cite{Krstic2009TAC},
   heat equation  \cite{Krstic2009SCL}   as well as  the   Schr\"{o}dinger equation
\cite{WangJMSCL2013},
       in a unified way.
    More importantly, the more complicated problem that stabilize the PDEs  with ODE actuator dynamics can still be addressed effectively.
 To validate the effectiveness of the developed  method, the proposed scheme of controller design  will be applied to stabilization of
 ODE-transport  cascade and heat-ODE cascade  in sections \ref{Se.6}  and \ref{Se.7}, respectively.

 \begin{remark}\label{Re2020572112}
Another interesting question is to extend   Theorems \ref{Th201912021844} and \ref{Th20205271541} to
the case where  both $A_1$ and $A_2$ are unbounded.
There are still many  difficulties  to achieve this problem  in the  general abstract framework.
 One of the reasons is that the general Sylvester equation with unbounded operators is hard to be solved.
In addition, the proof of
   well-posedness and exponential stability is also  difficult. However,
   the main idea of the developed approach  is still helpful to  the actuator dynamics compensation
    with  the unbounded   $A_1$ and $A_2$. This will be considered in the
   third  paper  \cite{FPartDelay} of this series works.
 \end{remark}

\section{ODEs with input  delay}\label{Se.6}
In this section, we apply the proposed approach  to  the input   delay compensation for  ODEs.
Consider the following
  linear   system in the  state space $X_1=\R^n$:
   \begin{equation} \label{2018921217}
 \left.\begin{array}{l}
\dot{x}_1(t) = A_1x_1(t)+B_1 u(t-\tau),\ \  \tau>0,
\end{array}\right.
\end{equation}
where    $A_1\in \mathcal{L}( X_1)$
is the system operator,
$B_1\in \mathcal{L}(\R,X_1)$ is the  control operator, and
  $u:  [-\tau,\infty) \rightarrow \R $ is the  scalar  control  that is  delayed by   $\tau$ units of time.
  It should be pointed out that the input delay  compensation   problem \dref{2018921217}
  has been considered via  many approaches such as
    the spectrum
assignment  approach in \cite{KwonPearson1980TAC}, the ``reduction approach" in \cite{Artstein1982TAC}  and the PDE backstepping method in \cite{SmyshlyaevKrstic2008SCL}. In this section, we
   re-consider this problem and show our differences with other approaches.
     Let
 \begin{equation} \label{2018921236}
 w(x,t)=u(t-x) ,\ \ x\in[0,\tau ],\ \   t\geq 0 .
 \end{equation}
Then,  system \dref{2018921217}
 can be written  as
 \begin{equation} \label{201712252004}
\left\{\begin{array}{l}
\disp  \dot{x}_1(t)=A_1x_1(t) + B_1w(\tau ,t),\ \ t>0,\crr
\disp    w_t(x,t)+w_x(x,t)=0,\ \ x\in[0,\tau],\ \  t>0,\crr
\disp w(0,t)=u(t),\ \ t\geq0,
\end{array}\right.
\end{equation}
which clearly shows why   the time-delay is infinite-dimensional  and \dref{201712252004}
now is  delay free.  In order to write system \dref{201712252004}
into  the abstract form  \dref{201911101138}, we
define $ A_2  :D( A_2 )\subset L^2[0,\tau]\to L^2[0,\tau]$ by
\begin{equation} \label{20191150018}
 \left.\begin{array}{l}
 A_2  f=-f' ,\ \ \forall\ f\in D( A_2 )=\left\{f\in H^1(0,\tau)\ |\ f (0)=0\right\},
\end{array}\right.
\end{equation}
and define $B_2q=q\delta(\cdot)$ for any $q\in\R$, where $\delta(\cdot)$ is the Dirac distribution.
 System  \dref{201712252004} can be written  as  the abstract form:
\begin{equation} \label{201911511293}
 \left\{\begin{array}{l}
 \dot{x}_1(t)=A_1x_1(t)+B_1C_2w(\cdot,t),\crr
w_t(\cdot,t)= A_2  w (\cdot,t)+B_2u(t),
\end{array}\right.
\end{equation}
where $C_2 f= f(\tau)$ for all $f\in D(A_2)$. Define the vector-valued function $\Phi:[0,\tau]\to\R^n$ by
$\Phi(x)=(\Phi_1(x),\Phi_2(x),\cdots,\Phi_n(x))^{\top} $  for any $x\in[0,\tau]$,
where
   $\Phi_j\in  L^2[0,\tau]$  will be determined later, $j=1,2,\cdots,n$.
 Suppose that  the solution of  Sylvester equation
  \dref{201911301901} takes the form \dref{2020529729}.
Then,   $\Phi(\cdot)$ satisfies
  \begin{equation} \label{2020821450}
 \dot{\Phi}(x)=A_1\Phi(x),\ \  \Phi(\tau)= B_1.
 \end{equation}
 We solve \dref{2020821450} to obtain the  solution   of Sylvester equation
 \dref{201911301901}
 \begin{equation} \label{201912071844}
 Sf=\int_0^{\tau}e^{A_1(x-\tau)}B_1f(x)dx,\ \ \ \forall\ f\in  L^2[0,\tau].
 \end{equation}
  As a result,
   \begin{equation} \label{201912071844818}
     SB_2q=q\int_0^{\tau}e^{A_1(x-\tau)}B_1\delta(x)dx =e^{-A_1\tau}B_1q , \ \ \forall\  q\in\R.
     \end{equation}
If there exists a  $K \in \mathcal{L}(   X_1, \mathbb{R})$
 such that
  $A_1+ B_1K $ is Hurwitz, then
   the operator  $A_1+ e^{-A_1 \tau}B_1K e^{ A_1 \tau}$ is also Hurwitz due to the invertibility of
   $e^{-A_1 \tau}$.
 Since $e^{A_2t}$ is exponentially stable   already,  by \dref{201912031020}, the controller is then designed as
  \begin{equation} \label{20191251106trans}
 u(t)=
   K_1\int_0^{\tau} e^{A_1(x-\tau)}B_1w (x,t)dx+K_1{x}_1(t),\ K_1=Ke^{ A_1 \tau},
\end{equation}
which leads to the closed-loop system:
 \begin{equation} \label{201911051108trans}
 \left\{\begin{array}{l}
\disp \dot{x}_1(t) = A_1  x_1(t)+B_1  w(\tau,t),\ \ t>0, \crr
\disp    w_t(x,t)+w_x(x,t)=0,\ \ x\in[0,\tau],  \ \ t> 0,\crr
\disp w(0,t)=  K \int_0^{\tau} e^{A_1x}B_1 w (x,t)dx+Ke^{ A_1 \tau} {x}_1(t),\ \ t\geq0.
\end{array}\right.
\end{equation}
By \dref{2018921236},  the controller \dref{20191251106trans} can be rewritten as
 \begin{equation} \label{201912091111}
 \begin{array}{ll}
 u(t) \disp
  \disp =K\left[e^{A_1\tau}x_1(t)+\int_{t-\tau}^{t}e^{A_1(t-\sigma)}B_1u(\sigma)d\sigma\right],\ \ t\geq\tau,
 \end{array}
\end{equation}
which  is  the same   as those
obtained by  the spectrum
assignment  approach in \cite{KwonPearson1980TAC}, the ``reduction approach" in \cite{Artstein1982TAC}  and the PDE backstepping method in \cite{SmyshlyaevKrstic2008SCL}.

  It is seen that in our approach, we never need the target system as that by the  backstepping approach.
  This avoids the possibility that when the target system  is not  chosen properly, there is no state feedback control
  and even if the target system is good enough, there is difficulty in solving  PDE kernel equation for the  backstepping transformation.
  Another  advantage of the  proposed  approach   is that we never construct
  the  Lyapunov   function   in the stability
analysis, which avoids another difficulty of construction of the Lyapunov function.
 Finally,  we point out that the  proposed approach  is still working  for the  unbounded operator $A_1$,
 which will be considered  in detail  in the
  third  paper  \cite{FPartDelay} of this series works.

\section{Heat equation with ODE dynamics}\label{Se.7}

In this section, we consider the stabilization of  an  unstable  heat equation with  $m$-dimensional  ODE actuator dynamics as follows:
 \begin{equation} \label{201912291553}
 \left\{\begin{array}{l}
\disp w_t(x,t)=w_{xx}(x,t)+ \mu w(x,t),\ \ x\in(0,1), \ t>0,\crr
\disp  w(0,t)=0,\ w_x(1,t)=C_2x_2(t),\  \ t\geq0,\crr
\disp \dot{x}_2(t) = A_2  x_2(t)+B_2 u(t),\ \ t>0,
\end{array}\right.
\end{equation}
where $w(\cdot,t)$ is the state of the heat system, $\mu>0$, $A_2\in  \R^{m\times m} $ is the system matrix of  the actuator dynamics,
$C_2\in \mathcal{L}(\R^{m},\R)$ represents the connection,  $B_2\in \mathcal{L}(\R,\R^m)$  is the control operator  and $u(t)$ is the control.
We assume without loss of the generality
that $A_2$ is  Hurwitz. Compared with the
stabilization of  finite-dimensional systems through infinite-dimensional dynamics in existing literature,
 stabilization of   infinite-dimensional unstable system through  finite-dimensional dynamics is a difficult
 problem  and the corresponding result is very  scarce.

Define the operator  $ A_1  :D( A_1 )\subset L^2[0,1]\to L^2[0,1]$ by
\begin{equation} \label{201912311648}
 \left\{\begin{array}{l}
 A_1  f(\cdot)=f''(\cdot)+ \mu f(\cdot),\ \ \forall\ f\in D( A_1 ),\crr
D( A_1 )=\left\{f\in H^2(0,1)\mid f(0)= f'(1)=0\right\},
\end{array}\right.
\end{equation}
and  the operator $B_1:\mathbb{R}:\to [D(A_1^*)]'$ by
$B_1 c=c\delta(\cdot-1)$ for any $c\in \mathbb{R}$, where
$\delta(\cdot)$ is the Dirac distribution.
With these operators at hand,   system \dref{201912291553} can be written as the abstract form:
\begin{equation} \label{201912311647}
 \left\{\begin{array}{l}
 w_t(\cdot,t) =A_1w (\cdot,t)+B_1C_2x_2(t),\crr
\dot{x}_2(t)= A_2 x_2 (t)+B_2u (t).
\end{array}\right.
\end{equation}
Let $\Psi(\cdot)=(\Psi_1(\cdot),\Psi_2(\cdot),\cdots,\Psi_m(\cdot))^{\top} $ be a vector-valued function over
$[0,1]$, where
$\Psi_j\in L^2[0,1]$ will  be determined later, $j=1,2,\cdots,m$.
Suppose   that the solution of     Sylvester equation \dref{201911301901} takes the form
  \dref{2020529739}.  Then, a straightforward computation shows that   $\Psi(\cdot)$ satisfies
  \begin{equation}\label{2019122291634}
\left.\begin{array}{l}
\disp \Psi''(x) =(A_2^*- \mu)\Psi(x) ,\ \
 \disp \Psi(0)=0,\ \ \Psi'(1)=-C_2^{\top}.
 \end{array}\right.
\end{equation}
Solve system \dref{2019122291634} to obtain the solution
 \begin{equation}\label{2019122291644}
\left.\begin{array}{l}
\disp \Psi(x)=-\sinh Gx(G\cosh G)^{-1}C_2^{\top} ,\ \ G^2=A_2^*- \mu,\ x\in[0,1].
 \end{array}\right.
\end{equation}
By \dref{2020529739},  the solution   of   Sylvester equation \dref{201911301901} is
 \begin{equation} \label{201912311728}
 Sh=\langle h, \Psi(\cdot)\rangle_{\R^m}=\sum_{i=1}^{m}\Psi_i(\cdot)h_i,\ \ \ \forall\ h=(h_1,h_2,\cdots,h_m)^{\top}\in \R^m.
 \end{equation}
 According to the scheme of the compensator design  at the end of section \ref{Se.5},
 we need to stabilize system $(A_1,SB_2)$ which is associated with the following system:
 \begin{equation} \label{201912311823}
 \left\{\begin{array}{ll}
\disp    z _t(x,t)= z _{xx}(x,t)+ \mu z (x,t)+b(x)u(t),\ \ x\in(0,1),\ t>0, \crr
\disp   z (0,t)=  z _x(1,t)=0,\ \ t\geq0,
\end{array}\right.
\end{equation}
where $\mu>0$,  $z(\cdot,t)$ is the new state,   $u(t)$ is the control and
\begin{equation} \label{2020881745}
 b (\cdot)q = SB_2 q=q\sum_{i=1}^{m}\Psi_i(\cdot)b_{2i} ,\ \ B_2=(b_{21},b_{22},\cdots,b_{2m})^{\top},\ \ \forall\ q\in \R .
\end{equation}

Inspired by   \cite{CoronTrelat2004SICON,PrieurandTrelat2019TAC,Russell1978SIAMReview},   system \dref{201912311823} can be stabilized by   the  finite-dimensional spectral truncation technique.
 Let
\begin{equation}\label{wxh201912303}
\phi_n(x)=\sqrt{2}  \sin \sqrt{\lambda_n}x , \ \ \lambda_n=\left(n-\frac{1}{2}\right)^2\pi^2, \ \ x\in[0,1],\ \ n\geq 1.
\end{equation}
Then,  $\{\phi_n(\cdot) \}_{n=1}^{\infty}$ forms an orthonormal   basis for  $L^2[0,1]$ and satisfies
\begin{equation}\label{wxh201912302}\left.\begin{array}{l}
\phi_n''(x)=-\lambda_n\phi_n(x),\ \
 \phi_n(0)=\phi'_n(1)=0,\ \ \  n=1,2,\cdots.
\end{array}\right.\end{equation}
The function $b(\cdot)$ and the solution  $z(\cdot,t)$ of \dref{201912311823}
can be represented as
\begin{equation}\label{2020881457}\left.\begin{array}{l}
\disp b(\cdot)=\sum\limits_{n=1}^{\infty}b_n \phi_n(\cdot),\ \ b_n =\displaystyle \int_{0}^{1}b(x)\phi_n(x)dx, \ \ n=1,2,\cdots
\end{array}\right.\end{equation}
and
 \begin{equation}\label{2020881609}\left.\begin{array}{l}
\disp  z (\cdot,t)=\sum\limits_{n=1}^{\infty}z_n(t)\phi_n(\cdot),  \ \ z_n(t)=\displaystyle \int_{0}^{1} z (x,t)\phi_n(x)dx,\ \  n=1,2,\cdots.
\end{array}\right.\end{equation}
 By   \dref{201912311823},  \dref{wxh201912303}  and \dref{wxh201912302}, it follows   that
\begin{equation}\label{wxh201912305} \begin{array}{rl}
\dot{z}_n(t)=&\displaystyle \int_{0}^{1} z _t(x,t)\phi_n(x)dx
=\displaystyle \int_{0}^{1}\left[ z _{xx}(x,t)+  \mu z (x,t)+b(x)u(t)\right]\phi_n(x)dx\crr
=&(-\lambda_n+  \mu )z_n(t)+b_nu(t).
\end{array} \end{equation}
If we choose   the  integer  $N$  large enough such  that
\begin{equation}\label{201912301959}
 (-\lambda_n+ \mu)<0,\ \ \forall\ n>N,
 \end{equation}
then,   $z_n(t)$ is stable for all $n>N$.  It is therefore sufficient   to consider
$z_n(t)$  for  $n\leq N$, which satisfy the following finite-dimensional system:
\begin{equation}\label{201912302011}
\dot{Z}_N(t)=\Lambda_NZ_N(t)+B_Nu(t),\ \ Z_N(t)=(z_1(t),\cdots, z_{N}(t))^\top,
\end{equation}
where $\Lambda_N$  and $B_N$ are defined by
\begin{equation}\label{2020871455}\left\{\begin{array}{l}
\Lambda_N={\rm diag} (-\lambda_1+ \mu,\cdots,-\lambda_{N}+ \mu) ,\crr
B_N=\left(b_1,b_2,\cdots,b_{N} \right)^\top .
\end{array}\right.\end{equation}
In this way, the stabilization of system \dref{201912311823} amounts to stabilizing
the finite-dimensional system
\dref{201912302011}. If
 there exists an $L_N=(l_1,l_2,\cdots,l_N)\in \mathcal{L}(\R^N,\R)$ such that  $\Lambda_N{+}B_NL_N$
 is Hurwitz, then it follows from
    Lemma \ref{Th201912302046} in Appendix
 that
   the operator $A_1+SB_2K_N$  generates an exponentially stable $C_0$-semigroup on $L^2[0,1]$, where
   $K_N\in \mathcal{L}(L^2[0,1],\R)$ is given by
   \begin{equation}\label{2020881654}
 K_N: f\to   \int_{0}^{1}f(x) \left[\sum_{k=1}^{N}l_k\phi_k(x)\right] dx,\ \ \forall\ f\in L^2[0,1].
\end{equation}
  Taking  \dref{201911281134infinite} and \dref{201912311728} into account, the controller of system \dref{201912291553} can be designed as
 \begin{equation}\label{201912301847}
u (t)= K_N[Sx_2(t)+w(\cdot,t)]=K_N\left[\langle \Psi(\cdot),x_2(t)\rangle_{\R^m}+  w(\cdot,t)\right],
\end{equation}
which leads to the closed-loop system:
\begin{equation} \label{201912302033}
 \left\{\begin{array}{l}
\disp w_t(x,t)=w_{xx}(x,t)+ \mu w(x,t), \ \ x\in(0,1),\ \ t>0,\crr
\disp  w(0,t)=0,\ w_x(1,t)=C_2x_2(t),\ \ t\geq0,\crr
\disp \dot{x}_2(t) = A_2  x_2(t)+B_2  K_N[ \langle \Psi(\cdot),x_2(t)\rangle_{\R^m}+w(\cdot,t)],\ \
t>0,
\end{array}\right.
\end{equation}
where $\Psi(\cdot) $  and $K_N$ are given by \dref{2019122291644} and \dref{2020881654}, respectively.

\begin{theorem}\label{wxhTh201912302046}
Suppose that system \dref{201912291553} is approximately controllable, $A_2$ is Hurwitz,
 $\sigma(A_1)\cap \sigma(A_2)=\emptyset$,  and the constants $ \mu $ and $N$ satisfy \dref{201912301959}.
Then, there exists an $L_N=(l_1,l_2,\cdots,l_N)\in \mathcal{L}(\R^N,\R)$ such that,
for  any initial state  $(w(\cdot,0), x_2(0)) ^\top \in L^2[0,1]\times\R^m$, the closed-loop system \dref{201912302033}
admits a unique solution
$(w, x_2) ^ {\top}\in C([0,\infty);  L^2[0,1]\times\R^m)$ which
decays to zero  exponentially    in $L^2[0,1] \times \R^m$ as  $t\to \infty$.
\end{theorem}
\begin{proof}
We first show that $\Psi(\cdot)$ defined by \dref{2019122291644} makes sense under the assumptions. Indeed, \dref{2019122291644}  can be rewritten as
\begin{equation}\label{20208131755}
\left.\begin{array}{l}
\disp \Psi(x)=-x\mathcal{G}(xG)( \cosh G)^{-1}C_2^{\top} , \ x\in[0,1],
 \end{array}\right.
\end{equation}
where
\begin{equation}\label{20181141103}
 \mathcal{G}(s)=\left\{\begin{array}{ll}
   \disp \frac{\sinh s}{s} ,& s\neq0,s\in \mathbb{C},\\
   1, & s=0.
 \end{array}\right.
   \end{equation}
   By  \cite[Definition 1.2, p.3]{Higham2008book},   both $\mathcal{G}(xG)$  and $\cosh G $  are always well defined. It suffices to prove that  $\cosh G $ is invertible.
Since  $\sigma(A_1) \cap \sigma(A_2)=\emptyset$ and $A_1=A_1^*$,
  a simple computation shows that
  \begin{equation}\label{20208141536}
 \sigma(G^2  )   \cap \sigma(A_1^*-\mu  )=\emptyset,  \ \ G^2=A_2^*- \mu
  \end{equation}
   and
 \begin{equation}\label{208131822}
\sigma(A_1^*-\mu  )=  \left\{ -\left(n-\frac{1}{2}\right)^2\pi^2\ \Big{|}\ n\in \mathbb{N}\right\}.
  \end{equation}
 For any  $\lambda \in \sigma(G)$,  we have $\lambda^2\in \sigma(G^2  )$ and hence
  $\lambda^2 \notin \sigma(A_1^*-\mu  )$.
This implies that  $\lambda \notin \left \{  \left(n-\frac{1}{2}\right)\pi i\ |\ n\in \mathbb{Z}\right\}$ and hence $\cosh \lambda \neq 0 $.
Consequently,   $\cosh G$ is invertible. The function $\Psi(\cdot)$   is well defined.

 By a simple computation, the operator $S$ given by \dref{201912311728} and \dref{2019122291644}  solves
  the Sylvester equation
 \dref{201911301901} and    $SB_2$  given by \dref{2020881745} satisfies $SB_2\in \mathcal{L}(\mathbb{R},L^2[0,1])$.
Define $\A =\begin{pmatrix}
\tilde{A}_1&B_1C_{2 }\\0& {A}_2
\end{pmatrix}$  and
$  \mathcal{B}_2 =\begin{pmatrix}
 0 \\
 B_2
\end{pmatrix} \in \mathcal{L}(\R, L^2[0,1]\times\R^m) $. Then,
     $(\A, \mathcal{B}_2)$  is approximately controllable. Similarly to the proof of Corollary \ref{Co201912036820}, it follows from  Lemma   \ref{th201911262002} that the pair
    $(\mathbb{S}\A \mathbb{S}^{-1},\mathbb{S}\mathcal{B}_2)=\left( \begin{pmatrix}
\tilde{A}_1&0\\0& {A}_2
\end{pmatrix}, \begin{pmatrix}
SB_2\\
B_2
\end{pmatrix}\right)$  is approximately controllable as well
    where the invertible
    transformation $\mathbb{S}$   is given by
   \begin{equation} \label{20205272026***}
   \left.\begin{array}{l}
\disp    \mathbb{S}  \left( f ,  x_2  \right)^{\top}= \left( f+S x_2,\ x_2  \right)^{\top},\ \ \forall\ ( f, {x}_2)^\top\in  L^2[0,1]\times \R^m.
 \end{array}\right.
\end{equation}
 Thanks to   the  block-diagonal structure of $\mathbb{S}\A \mathbb{S}^{-1}$ and
      Lemma \ref{Lm20205291214} in  Appendix,      system     $(A_1,  SB_2)$ is  approximately  controllable. By \dref{2020881745}, system  $(A_1,  b(\cdot))$ is  approximately  controllable as well.
      Since   $\{\phi_n(\cdot)\}_{n=1}^{\infty}$ defined  by \dref{wxh201912303}  forms an orthonormal   basis for  $L^2[0,1]$,   we then conclude   that
         \begin{equation}\label{2020881648}
 b_n =\displaystyle \int_{0}^{1}b(x)\phi_n(x)dx\neq0,\ \ n=1,2,\cdots,N,
\end{equation}
      which, together with  \dref{2020871455},  implies that the finite-dimensional linear system $(\Lambda_N,B_N)$ is controllable.
    As a result,  there exists an $L_N=(l_1,l_2,\cdots,l_N)\in \mathcal{L}(\R^N,\R)$
      such that
 $\Lambda_N{+}B_NL_N$
 is Hurwitz.
       By
    Lemma \ref{Th201912302046},
       the operator $  A_1+SB_2K_N$ generates an exponentially stable $C_0$-semigroup
  on $L^2[0,1]$.
 This  completes the proof   by
  Theorem \ref{Th20205271541}.
\end{proof}

\section{Numerical simulations}\label{Se.8}
 In this section, we carry out  some simulations for system \dref{201912302033} to validate our theoretical  results.
We choose
 \begin{equation}\label{wxh202068706}
A_2 = \begin{pmatrix}
  -1& 0\\
 0 & -2
\end{pmatrix},\ B_2=\begin{pmatrix}
  1\\
1
\end{pmatrix},\  \ C_2=(1, 1),\  \ \mu=10.
\end{equation}
It is easily  to check that the assumptions in Theorem \ref{wxhTh201912302046}  are fulfilled with $N=1$.
The    initial states of  system \dref{201912302033} are chosen as
$x_2(0)=(1,1)^\top$  and $  w(x,0)=\sin \pi x$  for any $ x\in[0,1]$.
 The finite difference scheme is adopted in discretization.
The time  and   space steps  are
taken as  $4\times 10^{-5}$ and $ 10^{-2}$, respectively. The numerical results
are programmed in Matlab.
We assign  the poles  to get  the gains
$\Lambda_N=7.5326$, $B_N=0.3130$ and  $L_N=-30.4541$, which yield
 $\sigma(\Lambda_N+B_NL_N)=\{-2\}$.


The solution of the open-loop system  \dref{201912291553} with $u=0$ is plotted in
Figure~\ref{Fig1} (a) and (b) which show that the control free system is indeed unstable.
 The trajectory of state feedback law $u(t)$ is plotted  in  Figure~\ref{Fig1} (c).
 The
state $w(\cdot,t)$ of  the  closed-loop system \dref{201912302033}
is plotted in
Figure~\ref{Fig2} (a) and the state $x_2(t)$  is plotted in Figure~\ref{Fig2} (b).
Comparing Figure \ref{Fig1} with Figure   \ref{Fig2},
it is found  that  the proposed approach is very
  effective and the  controller    is   smooth.

\begin{figure}[!htb]\centering
\subfigure[  $ w(x,t)$.]
 {\includegraphics[width=0.32\textwidth]{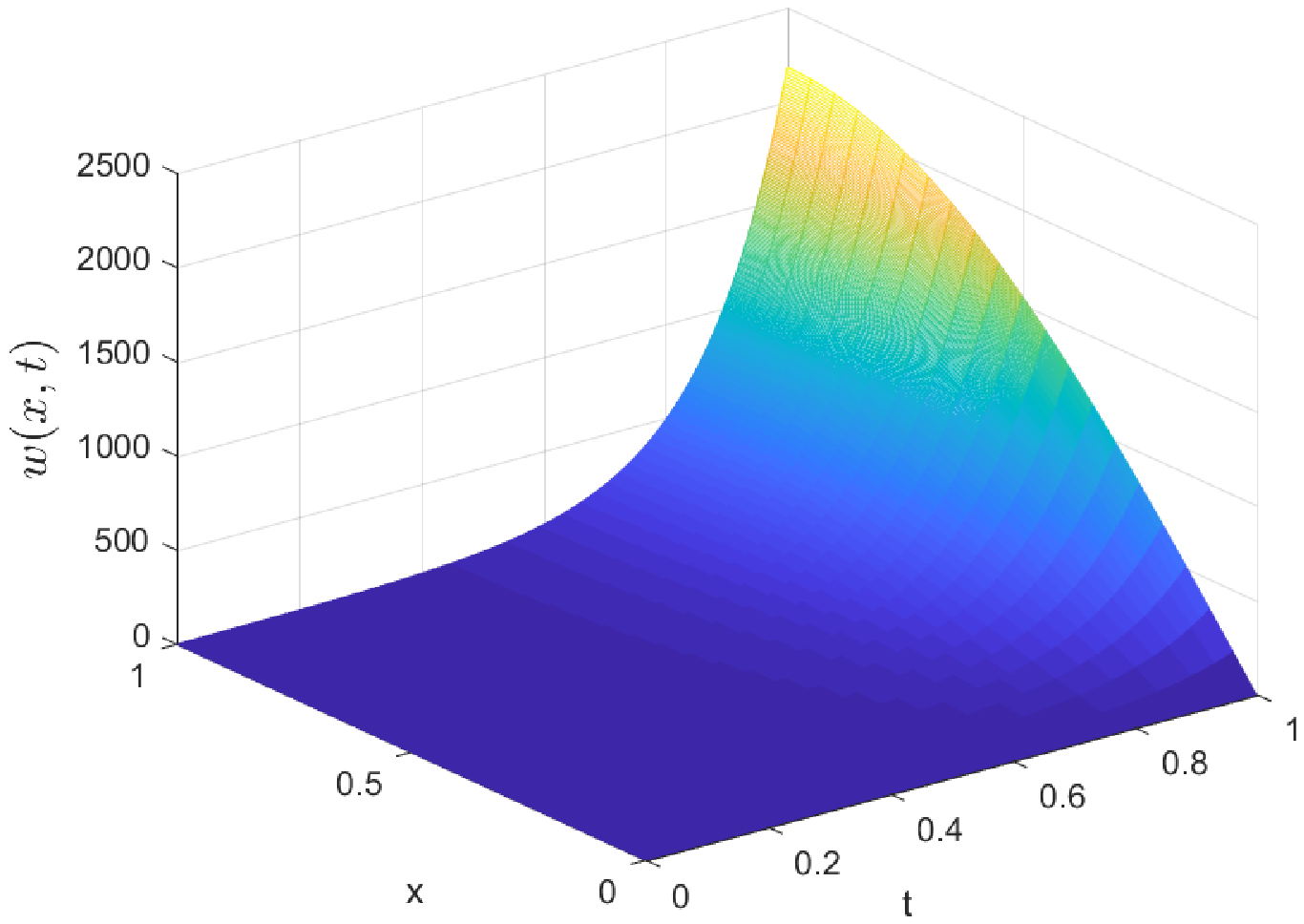}}
\subfigure[  $x_2(t)=(x_{21}(t),x_{22}(t))^\top$.]
 {\includegraphics[width=0.32\textwidth]{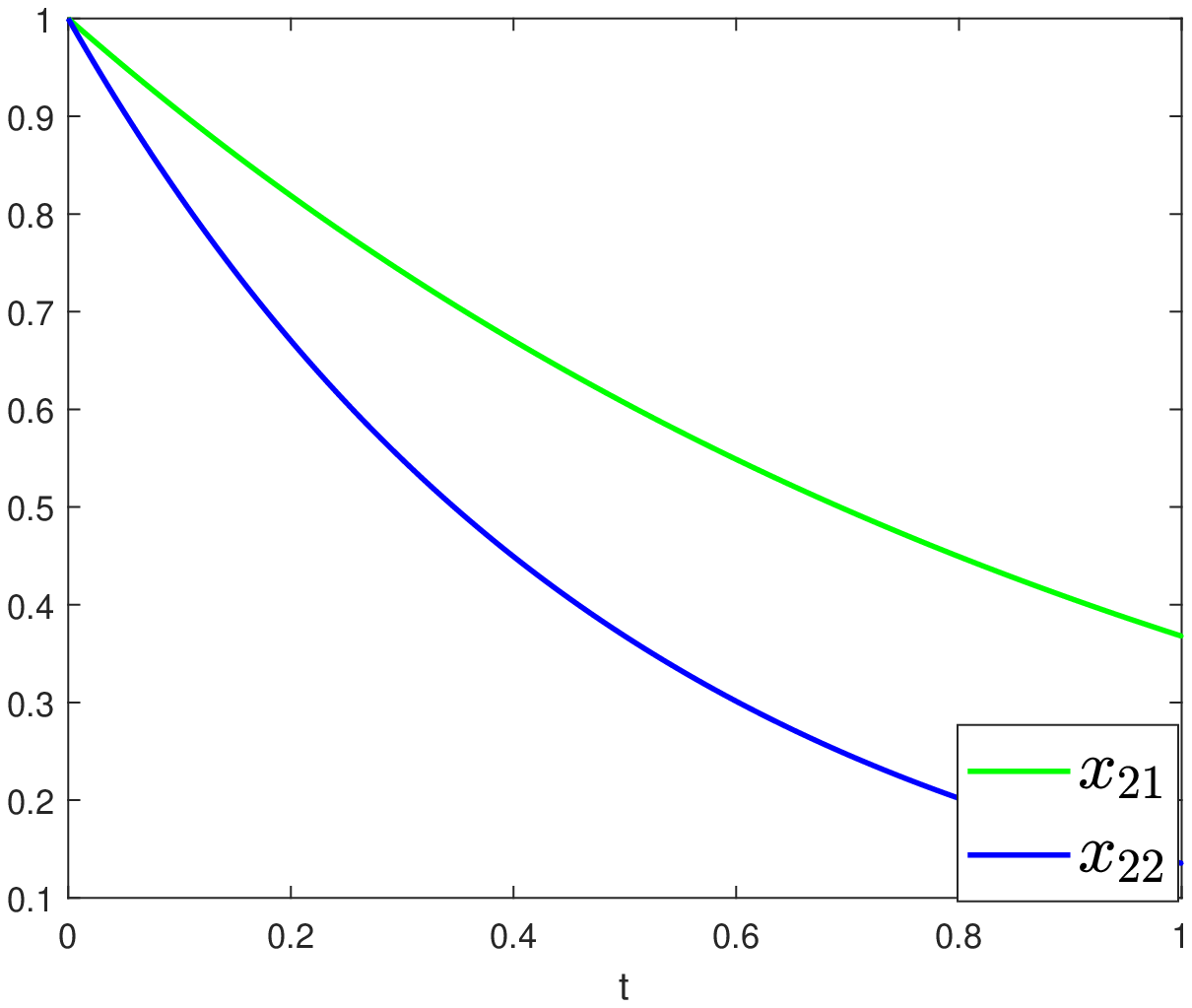}}
  \subfigure[ Controller.]
 {\includegraphics[width=0.32\textwidth]{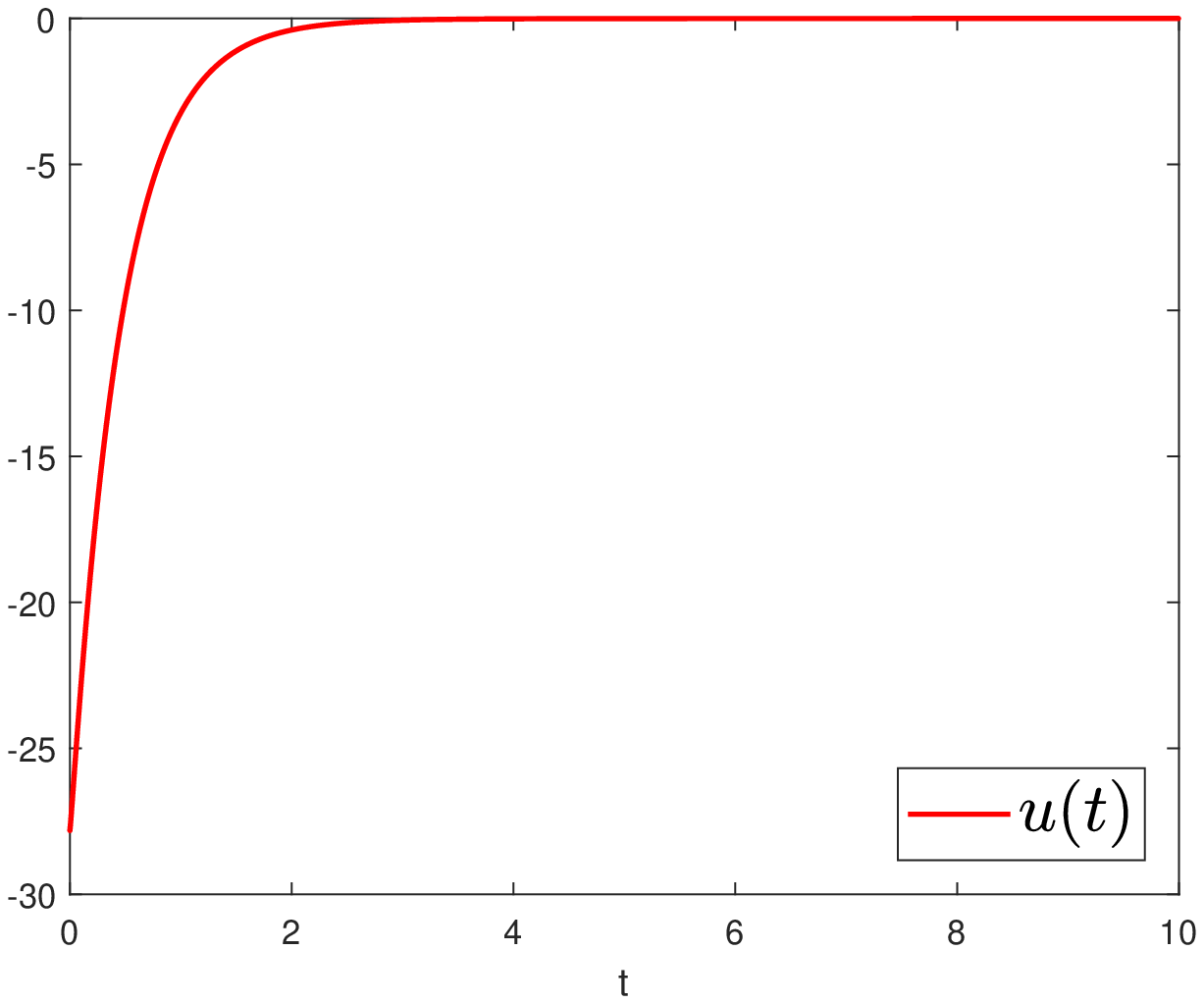}}
\caption{ Solution of open-loop    \dref{201912291553}  and feedback law.}\label{Fig1}
\end{figure}

\begin{figure}[!htb]\centering
\subfigure[$w(x,t)$.]
 {\includegraphics[width=0.48\textwidth]{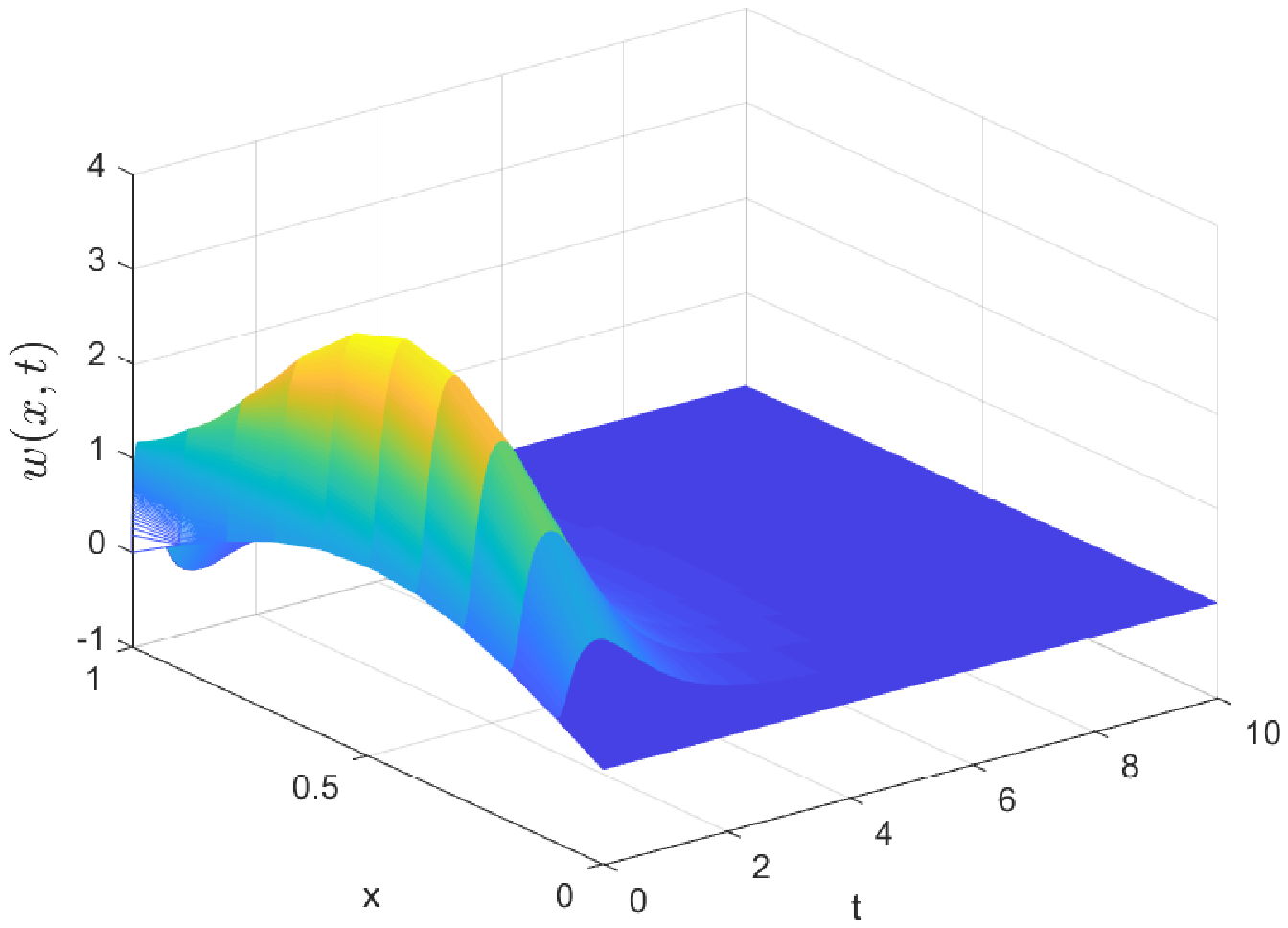}}
 \subfigure[ $x_2(t)=(x_{21}(t),x_{22}(t))^\top$.]
 {\includegraphics[width=0.48\textwidth]{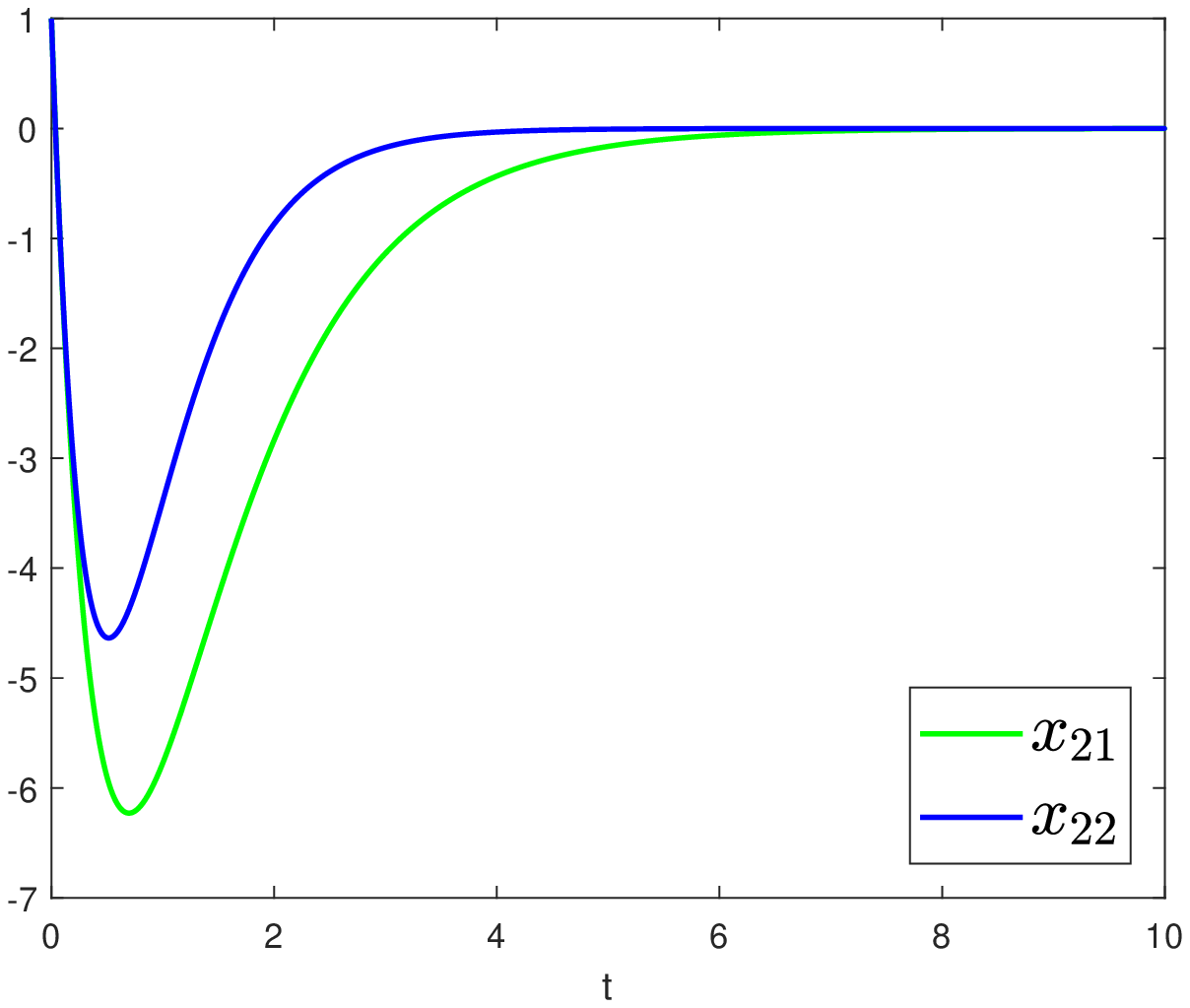}}
 \caption{Solution of closed-loop \dref{201912302033}.}\label{Fig2}
\end{figure}

\section{Conclusions}\label{Se.9}
In this paper, we develop a systematic  method to  compensate    the   actuator dynamics
  dominated  by general abstract  linear systems.  A  scheme of
full state feedback law design is proposed.
As a result,   a sufficient condition of the existence of compensator for
ODE with PDE actuator dynamics is obtained and
the existing results about stabilization of ODE with
 actuator dynamics dominated by the
   transport  equation  \cite{SmyshlyaevKrstic2008SCL},   wave equation  \cite{Krstic2009TAC},
   heat equation  \cite{Krstic2009SCL}   as well as  the   Schr\"{o}dinger equation
\cite{WangJMSCL2013}
       can be treated in a unified way.
       More importantly, the more complicated problem that stabilize the infinite-dimensional system
through   finite-dimensional actuator dynamics can still be addressed effectively.
 We present  two examples to demonstrate the effectiveness of the proposed  approach.
One is on  input   delay compensation for ODE system and another is for
unstable heat equation with ODE actuator dynamics.

It should be pointed out that the proposed  approach in Theorems \ref{Th201912021844}  and \ref{Th20205271541} is not limited to
the examples considered in Sections
\ref{Se.6} and \ref{Se.7}.
In  \cite{FengWubeam}, it has been applied to  the stabilization of ODEs with
actuator dynamics dominated by Euler-Bernoulli beam equation.
More importantly, the approach opens up a new road  leading to the stabilization of
cascade systems particularly for those systems which consist of  ODE and multi-dimensional PDE.

Furthermore, the main idea of the approach  is  still applicable to the stabilization of
PDE-PDE cascade systems like those arising from
    PDEs with input   delay.
 This  will  be considered   in the
  third paper  \cite{FPartDelay} of this series works.
The present   paper focuses only on the full state feedback.
 After being investigated in the  next  paper
 \cite{FPart2} of this  series studies for  the state observer design  through  sensor dynamics,
  the output feedback will  become    straightforward by the separation principle of the  linear systems.

%

\section{Appendix}\label{Appendix}

\begin{lemma}\label{Th201912302046}
Let the operator $A_1$ be given by \dref{201912311648},   $SB_2$, $b(\cdot) $   be given by \dref{2020881745} and    $\phi_n(\cdot)$, $\lambda_n$ be given by \dref{wxh201912303}.
Suppose that the  integer $N$  satisfies  \dref{201912301959},  $\Lambda_N$, $B_N$ is  given by  \dref{2020871455}
and     there exists an $L_N=(l_1,l_2,\cdots,l_N)\in \mathcal{L}(\R^N,\R)$ such that
 $\Lambda_N{+}B_NL_N$
 is Hurwitz. Then,
the operator $ A_1+SB_2K_N $ generates an exponentially stable $C_0$-semigroup
 on $L^2[0,1]$,
where
$ K_N $ is  given by \dref{2020881654}.

\end{lemma}
\begin{proof}
Since $A_1$ generates   an analytic semigroup $e^{A_1t}$  on $L^2[0,1]$ and
$ SB_2K_N $ is bounded, it follows from \cite[Corollary 2.3, p.81]{Pazy1983Book} that
 $  A_1+SB_2K_N $  also  generates an analytic semigroup on $L^2[0,1]$.
The proof will be accomplished if we can show that  $\sigma( A_1+SB_2K_N )\subset \{ s\ | \ {\rm Re}(s)<0\}$. For any $\lambda\in \sigma( A_1+SB_2K_N )$, we consider the characteristic equation
$ (A_1+SB_2K_N) f=\lambda f $ with $f\neq 0 $.

When $f\in  {\rm Span}\{\phi_1,\phi_2,\cdots,\phi_N\} $,  set $f=\sum_{j=1}^Nf_j\phi_j$. The   characteristic equation
becomes
 \begin{equation}\label{201912312136}
 \sum_{j=1}^Nf_jA_1\phi_j + b\sum_{j=1}^N f_j K_N\phi_j= \sum_{j=1}^N\lambda f_j\phi_j.
\end{equation}
Since $A_1\phi_j=(-\lambda_j+  \mu)\phi_j$ and
 \begin{equation}\label{201912312139}
   K_N\phi_j=    \int_{0}^{1} \phi_j(x)   \left[\sum_{k=1}^{N}l_k\phi_k(x)\right]dx
  =l_j  ,\ \ j=1,2,\cdots,N,
\end{equation}
the equation \dref{201912312136} takes the  form
\begin{equation}\label{201912312141}
 \sum_{j=1}^Nf_j(-\lambda_j+ \mu)\phi_j + b\sum_{j=1}^N f_jl_j= \sum_{j=1}^N\lambda f_j\phi_j.
\end{equation}
Take the  inner product with $\phi_n $, $n=1,2,\cdots,N$ on equation \dref{201912312141}
 to obtain
\begin{equation}\label{wxh201912312146}
 f_n(-\lambda_n+ \mu)  +  b_n\sum_{j=1}^N f_jl_j=  \lambda f_n,\ \ n=1,2,\cdots,N,
\end{equation}
which, together with \dref{2020871455}, leads to
\begin{equation}\label{201912312148}
 (\lambda -\Lambda_N-B_NL_N) \begin{pmatrix}
                       f_1\\f_2\\\vdots\\f_N
                     \end{pmatrix} =0.
\end{equation}
Since $(f_1,f_2,\cdots,f_N)\neq 0$, we have
\begin{equation}\label{wxh201912312149}
{\rm Det}(\lambda -\Lambda_N-B_NL_N) =0.
\end{equation}
 Hence, $\lambda\in \sigma(\Lambda_N{+}B_NL_N) \subset\{ s\ | \ {\rm Re} (s) <0\}$ since $\Lambda_N{+}B_NL_N$
 is Hurwitz.

  When   $f\notin {\rm Span}\{\phi_1,\phi_2,\cdots,\phi_N\}$,   there exists a $j_0>N$ such  that $\displaystyle \int_{0}^{1}f(x)\phi_{j_0}(x)dx\neq 0$.
Take the  inner product with $\phi_{j_0}$ on equation $ (A_1+SB_2K_N) f=\lambda f$
 to get \begin{equation}\label{wxh2020322257}
(-\lambda_{j_0}+\mu) \int_{0}^{1}f(x)\phi_{j_0}(x)dx=  \lambda \int_{0}^{1}f(x)\phi_{j_0}(x)dx,
\end{equation}
which implies that   $ \lambda =-\lambda_{j_0}+\mu<0$. Therefore,  $\lambda\in \sigma( A_1+SB_2K_N )\subset \{ s\ | \ {\rm Re}(s)<0\}$.
 The proof is  complete.
 \end{proof}

\begin{lemma}\label{Lm20205291214}
 Let $A_j$ be the generator of a $C_0$-semigroup
 $e^{A_jt}$  on $X_j$, $j=1,2$. Suppose that  $U$ is the control space and
 $B_j\in \mathcal{L}(U,[D(A_j^*)]')$   is admissible for $e^{A_jt}$, $j=1,2$.
   Suppose  further that
 $\sigma(A_1)\cap \sigma(A_2)=\emptyset  $ and system
 $\left({\rm  diag}(A_1,A_2),
\begin{pmatrix}
B_1\\
B_2
 \end{pmatrix}\right)$
is approximately  controllable. Then, both $(A_1, B_1)$  and $(A_2, B_2)$ are approximately  controllable.
  \end{lemma}
\begin{proof}
  By assumption, system   $({\rm diag}(A_1^*,A_2^*),(B_1^*,  B_2^*) )$
 is approximately  observable.   We use the argument of    proof by contradiction.
   If  either system $(A_1, B_1)$  or  $(A_2, B_2)$  were   not  approximately  controllable,
   we assume  without loss of the generality that system
   $(A_1, B_1)$  were   not  approximately  controllable. Then,
   system $(A_1^*, B_1^*)$    would not  be   approximately  observable.
   Hence,
 there exists $0\neq x_{10}\in X_1$ such that
   $B_1^*e^{A_1^*t}x_{10}\equiv0$ on $[0,\tau]$ for some time
    $\tau>0$. As a result, $ (B_1^*, B_2^*)e^{{\rm diag}(A_1^*,A_2^*)t}(x_{10},0)^{\top}\equiv0$
over  $[0,\tau]$, that is,   ${\rm Ker} \left( (B_1^*,  B_2^*)e^{{\rm diag}(A_1^*,A_2^*)t}\right)\neq\{0\}$.
This contradicts to the fact that   system $({\rm diag}(A_1^*,A_2^*),  (B_1^*, B_2^*) )$
 is approximately  observable.  The proof is complete.
\end{proof}

\begin{thebibliography}{0}

\bibitem{Artstein1982TAC} Z. Artstein, Linear systems with delayed controls: A reduction, {\it IEEE Transactions on Automatic Control}, 27(1982), 869-879.






\bibitem{CoronTrelat2004SICON}  J.M. Coron and E. Tr\'{e}lat, Global steady-state controllability of one dimensional semilinear heat equations, {\it SIAM Journal on Control and Optimization}, 43(2004), 549-569.





\bibitem{FengGuoWu2019TAC} H. Feng, B.Z. Guo and X.H. Wu, Trajectory planning approach to output tracking for a 1-d wave equation,   {\it IEEE Transactions on Automatic Control}, 65(2020), 1841-1854.




\bibitem{FengGuo2017TAC} H. Feng and B.Z. Guo,  A new active disturbance rejection control to output feedback stabilization for a one-dimensional anti-stable wave equation with disturbance,  {\it IEEE Transactions on Automatic Control},  62(2017), 3774-3787.




    \bibitem{FPart2} H. Feng,  X.H. Wu  and B.Z. Guo, Dynamics compensation in  observation of abstract   linear systems, to be  submitted
        ({\it \color{blue} as  the second part of this series of studies}).



 \bibitem{FPartDelay} H. Feng, Delays compensations for regular linear systems, to be submitted    ({\it \color{blue} as  the third part of this series of studies}).

 \bibitem{FPart4}    H. Feng and B.Z. Guo,  Extended dynamics observer for linear systems with disturbance, to be submitted    ({\it \color{blue} as  the last  part of this series of studies}).





\bibitem{KwonPearson1980TAC} W.H. Kwon and A.E. Pearson,  Feedback stabilization of linear systems with delayed control, {\it  IEEE Transactions on Automatic Control},  25(1980),   266-269.









\bibitem{SmyshlyaevKrstic2008SCL} M. Krstic and A. Smyshlyaev,  Backstepping boundary control for first order hyperbolic PDEs and application to systems with actuator and sensor delays,  {\it Systems $\&$ Control Letters},  57(2008),   750-758.


\bibitem{Krstic2009SCL} M. Krstic, Compensating actuator and sensor dynamics governed by diffusion PDEs, {\it Systems $\&$ Control Letters}, 58(2009), 372-377.

\bibitem{Krstic2009TAC} M. Krstic, Compensating a string PDE in the actuation or in sensing path of an unstable ODE, {\it IEEE Transactions on Automatic Control}, 54(2009), 1362-1368.





\bibitem{ManitiusOlbrot1979TAC}  A.Z. Manitius and A.W. Olbrot,  Finite spectrum assignment problem for systems with
delays, {\it  IEEE Transactions on Automatic Control}, 24(1979), 541-553.


\bibitem{Higham2008book} N.J. Higham, {\it Functions of Matrices Theory and Computation}, SIAM,  Philadelphia, 2008.



%



 \bibitem{Lassi2014TAC} L. Paunonen, The role of exosystems in output regulation, {\it IEEE Transactions on Automatic Control}, 59(2014), 2301-2305.

 \bibitem{Lassi2014SIAM}  L. Paunonen and S. Pohjolainen, The internal model principle for systems with
unbounded control and observation, {\it SIAM Journal on Control and Optimization}, 52(2014), 3967-4000.

\bibitem{Sylvester1991}  V.Q. Ph\'{o}ng,  The operator equation $AX-XB = C$ with unbounded operators $A$
and $B$ and related abstract Cauchy problems, {\it Mathematische Zeitschrift}, 208(1991), 567-588.


\bibitem{PrieurandTrelat2019TAC} C. Prieur and  E. Tr\'{e}lat,  Feedback stabilization of a 1-d linear reaction-diffusion equation with delay boundary control, {\it IEEE Transactions on Automatic Control}, 64(2019), 1415-1425.



 \bibitem{Pazy1983Book} A. Pazy,  {\it Semigroups of Linear Operators and Applications
  to Partial Differential Equations},
  Springer-Verlag, New York, 1983.

  \bibitem{WangJMSCL2013} B. Ren, J.M. Wang  and M. Krstic, Stabilization of an ODE-Schr\"{o}dinger cascade,
 {\it Systems $\&$ Control Letters},   62(2013), 503-510.

 \bibitem{Rosenblum1956DMJ} M. Rosenblum, On the operator equation $BX-XA=Q$,
  {\it Duke Mathematical Journal},  23(1956), 263-270.

\bibitem {Russell1978SIAMReview} D.L. Russell, Controllability and stabilizability theory for linear partial
differential equations: recent progress and open questions, {\it SIAM Review}, 20(1978), 639-739.


\bibitem {Salamon1987}   D. Salamon, Infinite-dimensional systems with unbounded control and observation: a functional analytic approach, {\it Transactions of  American Mathematical Society}, 300(1987), 383-431.


\bibitem{Sustoa2010JFI} G.A. Sustoa and  M. Krstic, Control of PDE-ODE cascades with Neumann interconnections, {\it Journal of the Franklin Institute}, 347(2010), 284-314.

%
 \bibitem{Smith1959ISA}  O.J. M. Smith,  A controller to overcome dead time,
 {\it  ISA J}, 6(1959), 28-33.



 \bibitem{TangSCL2011} S. Tang and  C. Xie, State and output feedback boundary control for a coupled PDE-ODE system, {\it Systems $\&$ Control Letters},  60(2011), 540-545.

 \bibitem{Thomasbook2} J.W. Thomas, {\it Numerical Partial Differential Equations: Conservation Laws and Elliptic Equations}, Springer-Verlag, New York, 1999.

\bibitem{TucsnakWeiss2009book} M. Tucsnak and G. Weiss,  {\it Observation and Control for Operator Semigroups}, Birkh\"{a}user,  Basel, 2009.


\bibitem{Wang2015AUT} J.M. Wang, J.J. Liu, B. Ren and J. Chen.
Sliding mode control to stabilization of cascaded heat PDE-ODE
systems subject to boundary control matched disturbance,
{\it Automatica}, 52(2015), 23-34.


\bibitem{Weiss1989SICON}  G. Weiss,  Admissibility of unbounded control operators,
{\it SIAM Journal on Control and Optimization}, 27(1989), 527-545.


\bibitem{Weiss1994MCSS} G. Weiss, Regular linear systems with feedback, {\it Mathematics of Control, Signals, and Systems},
 7(1994),
  23-57.

\bibitem{Weiss1989} G. Weiss, Admissible observation operators for linear semigroups,
{\it  Israel Journal of  Mathematgics},  65(1989), 17-43.


\bibitem{Weiss1997TAC} G. Weiss and R. Curtain, Dynamic stabilization of regular linear systems, {\it IEEE Transactions on Automatic Control}, 42(1997),   4-21.


\bibitem{Weidmann1980book} J. Weidmann, {\it Linear Operators in Hilbert Spaces}, Springer-Verlag, New York, 1980.



 \bibitem{FengWubeam} X.H. Wu and  H. Feng, Exponential stabilization of ODE system with
 Euler-Bernoulli beam actuator dynamics,  {\it SCIENCE CHINA Information Sciences},  to appear.




 \bibitem{Zhou2017JDE}   H.C. Zhou and B.Z. Guo,  Unknown input observer design and  output feedback stabilization for multi-dimensional wave equation  with boundary control matched uncertainty, {\it Journal of Differential Equations},  263(2017),  2213-2246.

\end{thebibliography}
 \end{document}